%% file: main.tex
\documentclass{article}

\usepackage[final]{neurips_2022}

\usepackage[utf8]{inputenc} 
\usepackage[T1]{fontenc}    
\usepackage[hidelinks]{hyperref}       
\usepackage{url}            
\usepackage{booktabs}       
\usepackage{amsfonts}       
\usepackage{nicefrac}       
\usepackage{microtype}      
\usepackage{xcolor}         
\usepackage[export]{adjustbox}
\usepackage{wrapfig,amsthm,amsmath, textcomp}
\usepackage{float}

\newtheorem{proposition}{Proposition}

\DeclareMathOperator{\rmsd}{RMSD}

\title{Torsional Diffusion for\\Molecular Conformer Generation}

\newcommand*\samethanks[1][\value{footnote}]{\footnotemark[#1]}
\author{Bowen Jing,\thanks{Equal contribution. Correspondence to \texttt{\{bjing, gcorso\}@mit.edu}.}\;\;\textsuperscript{1}\; Gabriele Corso,\samethanks\;\;\textsuperscript{1}\;  Jeffrey Chang,\textsuperscript{2}\; Regina Barzilay,\textsuperscript{1}\; Tommi Jaakkola\textsuperscript{1} \\
       \textsuperscript{1}CSAIL, Massachusetts Institute of Technology \;\;\; \textsuperscript{2}Dept. of Physics, Harvard University\\
}

\input{math_commands.tex}

\usepackage{booktabs}
\usepackage{url}
\usepackage{amsthm, graphicx, todonotes}
\newtheorem{prop}{Proposition}
\usepackage[ruled,vlined]{algorithm2e}
\usepackage{lipsum}
\usepackage{todonotes,multirow,makecell}
\newcommand{\comment}[1]{}

\begin{document}

\maketitle

\begin{abstract}
    Molecular conformer generation is a fundamental task in computational chemistry. Several machine learning approaches have been developed, but none have outperformed state-of-the-art cheminformatics methods. We propose \emph{torsional diffusion}, a novel diffusion framework that operates on the space of torsion angles via a diffusion process on the hypertorus and an extrinsic-to-intrinsic score model. On a standard benchmark of drug-like molecules, torsional diffusion generates superior conformer ensembles compared to machine learning and cheminformatics methods in terms of both RMSD and chemical properties, and is orders of magnitude faster than previous diffusion-based models. Moreover, our model provides exact likelihoods, which we employ to build the first generalizable Boltzmann generator. Code is available at \url{https://github.com/gcorso/torsional-diffusion}.
\end{abstract}

\section{Introduction} \label{sec:introduction}

Many properties of a molecule are determined by the set of low-energy structures, called \emph{conformers}, that it adopts in 3D space. Conformer generation is therefore a fundamental problem in computational chemistry \citep{hawkins2017conformation} and an area of increasing attention in machine learning. Traditional approaches to conformer generation consist of metadynamics-based methods, which are accurate but slow \citep{pracht2020automated}; and cheminformatics-based methods, which are fast but less accurate \citep{hawkins2010conformer, riniker2015better}. Thus, there is growing interest in developing deep generative models to combine high accuracy with fast sampling. 

Diffusion or score-based generative models \citep{ho2020denoising, song2021score}---a promising class of generative models---have been applied to conformer generation under several different formulations. These have so far considered diffusion processes in \emph{Euclidean} space, in which Gaussian noise is injected independently into every data coordinate---either pairwise distances in a distance matrix  \citep{shi2021learning, luo2021predicting} or atomic coordinates in 3D \citep{xu2021geodiff}. However, these models require a large number of denoising steps and have so far failed to outperform the best cheminformatics methods.

We instead propose \emph{torsional diffusion}, in which the diffusion process over conformers acts only on the torsion angles and leaves the other degrees of freedom fixed. This is possible and effective because the flexibility of a molecule, and thus the difficulty of conformer generation, lies largely in torsional degrees of freedom \citep{axelrod2020geom}; in particular, bond lengths and angles can already be determined quickly and accurately by standard cheminformatics methods. Leveraging this insight significantly reduces the dimensionality of the sample space; drug-like molecules\footnote{As measured from the standard dataset GEOM-DRUGS \citep{axelrod2020geom}} have, on average, $n=44$ atoms, corresponding to a $3n$-dimensional Euclidean space, but only $m=7.9$ torsion angles of rotatable bonds. \looseness=-1

Torsion angle coordinates define not a Euclidean space, but rather an $m$-dimensional torus $\mathbb{T}^m$ (Figure \ref{fig:overview}, \emph{left}). However, the dimensionality and distribution over the torus vary between molecules and even between different ways of defining the torsional space for the same molecule. To resolve these difficulties, we develop an \emph{extrinsic-to-intrinsic} score model (Figure~\ref{fig:overview}, \emph{right}) that takes as {input} a 3D point cloud representation of the conformer in {Euclidean} space (extrinsic coordinates), and predicts as {output} a score on a {torsional} space \emph{specific to that molecule} (intrinsic coordinates). To do so, we consider a \emph{torsional score} for a bond as a geometric property of a 3D point cloud, and use $SE(3)$-equivariant networks to predict them directly for each bond.

Unlike prior work, our model provides exact likelihoods of generated conformers, enabling training with the ground-truth \emph{energy} function rather than samples alone. This connects with the literature on \emph{Boltzmann generators}---generative models which aim to sample the Boltzmann distribution of physical systems without expensive molecular dynamics or MCMC simulations \citep{noe2019boltzmann, kohler2021smooth}. Thus, as a variation on the torsional diffusion framework, we develop \emph{torsional Boltzmann generators} that can approximately sample the conditional Boltzmann distribution for unseen molecules. This starkly contrasts with existing Boltzmann generators, which are specific for the chemical system on which they are trained.

Our main contributions are:
\begin{itemize}
    \item We formulate conformer generation in terms of diffusion modeling on the hypertorus---the first demonstration of non-Euclidean diffusion on complex datasets---and develop an extrinsic-to-intrinsic score model that satisfies the required symmetries: $SE(3)$ invariance, torsion definition invariance, and parity equivariance.
    \item We obtain state-of-the-art results on the GEOM-DRUGS dataset \citep{axelrod2020geom} and are the first method to consistently outperform the established commercial software OMEGA \citep{hawkins2017conformation}. We do so using two orders of magnitude \emph{fewer} denoising steps than GeoDiff \citep{xu2021geodiff}, the best Euclidean diffusion approach.
    \item We propose torsional Boltzmann generators---the first Boltzmann generator based on diffusion models rather than normalizing flows and the first to be useful for a class of molecules rather than a specific system.
\end{itemize}

\begin{figure}[t]
    \centering
    \includegraphics[width=\textwidth]{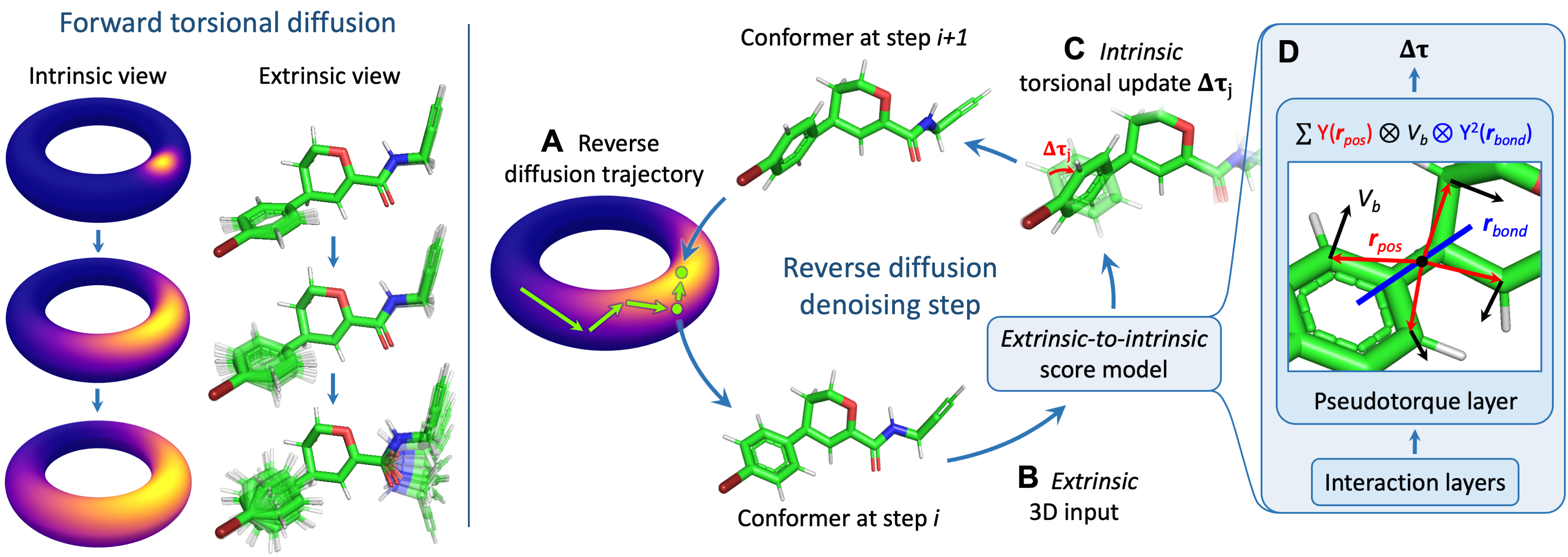}
    \caption{\textbf{Overview of torsional diffusion.} \emph{Left}: Extrinsic and intrinsic views of torsional diffusion (only 2 dimensions/bonds shown). \emph{Right}: In a step of reverse diffusion (\textbf{A}), the current conformer is provided as a 3D structure (\textbf{B}) to the score model, which predicts intrinsic torsional updates (\textbf{C}). The final layer of the score model is constructed to resemble a torque computation around each bond (\textbf{D}). $Y$ refers to the spherical harmonics and $V_b$ the learned atomic embeddings.}
    \label{fig:overview}
\end{figure}

\section{Background} \label{sec:background}

\paragraph{Diffusion generative models}
Consider the data distribution as the starting distribution $p_0(\mathbf{x})$ of a \emph{forward diffusion process} described by an Ito stochastic differential equation (SDE):
\begin{equation} \label{eq:forward}
    d\mathbf{x} = \mathbf{f}(\mathbf{x}, t)\; dt + g(t) \; d\mathbf{w}, \; \; t \in (0, T)
\end{equation}
where $\mathbf{w}$ is the Wiener process and $\mathbf{f}(\mathbf{x}, t), g(t)$ are chosen functions. With sufficiently large $T$, the distribution $p_T(\mathbf{x})$---the \emph{prior}---approaches a simple Gaussian. Sampling from the prior and solving the \emph{reverse diffusion}
\begin{equation} \label{eq:reverse}
    d\mathbf{x} = \left[\mathbf{f}(\mathbf{x}_t, t) - g^2(t)\nabla_\mathbf{x}\log p_t(\mathbf{x})\right]\; dt + g(t) \; d\mathbf{\bar w}
\end{equation} 
yields samples from the data distribution $p_0(\mathbf{x})$ \citep{anderson1982reverse,song2021score}. Diffusion, or score-based, generative models \citep{ho2020denoising,song2021score} learn the score $\nabla_\mathbf{x}\log p_t(\mathbf{x})$ of the diffused data with a neural network and generate data by approximately solving the reverse diffusion. The score of the diffused data also defines a \emph{probability flow ODE}---a continuous normalizing flow that deterministically transforms the prior into the data distribution \citep{song2021score}. We leverage the insight that, in many cases, this flow makes it possible to use diffusion models in place of normalizing flows and highlight one such case with the torsional Boltzmann generator.

Diffusion generative models have traditionally been used to model data on Euclidean spaces (such as images); however, \citet{de2022riemannian} recently showed that the theoretical framework holds with relatively few modifications for data distributions on compact Riemannian manifolds. The hypertorus $\mathbb{T}^m$, which we use to define torsional diffusion, is a specific case of such a manifold.

Several methods \citep{salimans2021progressive,vahdat2021score,nichol2021improved} have been proposed to improve and accelerate diffusion models in the domain of image generation. Among these, the most relevant to this work is \emph{subspace diffusion} \citep{jing2022subspace}, in which the diffusion is progressively restricted to linear subspaces. Torsional diffusion can be viewed in a similar spirit, as it effectively restricts Euclidean diffusion to a nonlinear \emph{manifold} given by fixing the non-torsional degrees of freedom.

\paragraph{Molecular conformer generation} The \emph{conformers} of a molecule are the set of its energetically favorable 3D structures, corresponding to local minima of the potential energy surface.\footnote{Conformers are typically considered up to an energy cutoff above the global minimum.} The gold standards for conformer generation are metadynamics-based methods such as CREST \citep{pracht2020automated}, which explore the potential energy surface while filling in local minima \citep{hawkins2017conformation}. However, these require an average of 90 core-hours per drug-like molecule \citep{axelrod2020geom} and are not considered suitable for high-throughput applications. Cheminformatics methods instead leverage approximations from chemical heuristics, rules, and databases for significantly faster generation \citep{lagorce2009dg,cole2018knowledge,miteva2010frog2,bolton2011pubchem3d,li2007caesar}; while these can readily model highly constrained degrees of freedom, they fail to capture the full energy landscape. The most well-regarded of such methods include the commercial software OMEGA \citep{hawkins2010conformer} and the open-source RDKit ETKDG \citep{landrum2013rdkit,riniker2015better}.

A number of machine learning methods for conformer generation has been developed \citep{xu2020learning,xu2021end,shi2021learning,luo2021predicting}, the most recent and advanced of which are GeoMol \citep{ganea2021geomol} and GeoDiff \citep{xu2021geodiff}. GeoDiff is a Euclidean diffusion model that treats conformers as point clouds $\mathbf{x} \in \mathbb{R}^{3n}$ and learns an $SE(3)$ equivariant score. On the other hand, GeoMol employs a graph neural network that, in a single forward pass, predicts neighboring atomic coordinates and torsion angles from a stochastic seed. 

\paragraph{Boltzmann generators} An important problem in physics and chemistry is that of generating independent samples from a Boltzmann distribution $p(\mathbf{x}) \propto e^{-E(\mathbf{x})/kT}$ with known but unnormalized density.\footnote{This is related to but distinct from conformer generation, as conformers are the local minima of the Boltzmann distribution rather than independent samples.} Generative models with exact likelihoods, such as normalizing flows, can be trained to match such densities \citep{noe2019boltzmann} and thus provide independent samples from an approximation of the target distribution. Such \emph{Boltzmann generators} have shown high fidelity on small organic molecules \citep{kohler2021smooth} and utility on systems as large as proteins \citep{noe2019boltzmann}. However, a {separate model} has to be trained for every molecule, as the normalizing flows operate on intrinsic coordinates whose definitions are specific to that molecule. This limits the utility of existing Boltzmann generators for molecular screening applications.

\section{Torsional Diffusion} \label{sec:torsional_diffusion}

Consider a molecule as a graph $G = (\mathcal{V}, \mathcal{E})$ with atoms $v \in \mathcal{V}$ and bonds $e \in \mathcal{E}$,\footnote{Chirality and other forms of stereoisomerism are discussed in Appendix~\ref{app:discuss:isomerism}.} and denote the space of its possible conformers $\mathcal{C}_G$. A conformer $C \in \mathcal{C}_G$ can be specified in terms of its \emph{intrinsic} (or internal) coordinates: local structures $L$ consisting of bond lengths, bond angles, and cycle conformations; and torsion angles $\boldsymbol{\tau}$ consisting of dihedral angles around freely rotatable bonds (precise definitions in Appendix~\ref{app:def}). We consider a bond \emph{freely rotatable} if severing the bond creates two connected components of $G$, each of which has at least two atoms.\footnote{Notably, this counts double bonds as rotatable. See Appendix~\ref{app:discuss:isomerism} for further discussion.}  Thus, torsion angles in cycles (or rings), which cannot be rotated independently, are considered part of the local structure $L$.

\emph{Conformer generation} consists of learning probability distributions $p_G(L, \boldsymbol{\tau})$. However, the set of possible stable local structures $L$ for a particular molecule is very constrained and can be accurately predicted by fast cheminformatics methods, such as RDKit ETKDG \citep{riniker2015better} (see Appendix~\ref{app:discuss:rdkit} for verification). Thus, we use RDKit to provide approximate samples from $p_G(L)$, and develop a diffusion-based generative model to learn distributions $p_G(\boldsymbol{\tau} \mid L)$ over torsion angles---conditioned on a given graph and local structure.

Our method is illustrated in Figure~\ref{fig:overview} and detailed as follows. Section \ref{sec:toroidal_diff} formulates diffusion modeling on the torus defined by torsion angles. Section \ref{sec:score_framework} describes the torsional score framework, Section \ref{sec:parity} the required symmetries, and Section \ref{sec:score_model} our score model architecture. Section \ref{sec:likelihood} discusses likelihoods, and Section \ref{sec:energy} how likelihoods can be used for energy-based training.

\subsection{Diffusion modeling on $\mathbb{T}^m$} \label{sec:toroidal_diff}

Since each torsion angle coordinate lies in $[0, 2\pi)$, the $m$ torsion angles of a conformer define a hypertorus $\mathbb{T}^m$. To learn a generative model over this space, we apply the continuous score-based framework of \cite{song2021score}, which holds with minor modifications for data distributions on compact Riemannian manifolds (such as $\mathbb{T}^m$) \citep{de2022riemannian}. Specifically, for Riemannian manifold $M$ let $\mathbf{x} \in M$, let $\mathbf{w}$ be the Brownian motion on the manifold, and let the drift $\mathbf{f}(\mathbf{x},t)$, score $\nabla_\mathbf{x} \log p_t(\mathbf{x})$, and score model output $\mathbf{s}(\mathbf{x}, t)$ be elements of the tangent space $T_\mathbf{x}M$. Then \eqref{eq:reverse} remains valid---that is, discretizing and solving the reverse SDE on the manifold as a \emph{geodesic random walk} starting with samples from $p_T(\mathbf{x})$ approximately recovers the original data distribution $p_0(\mathbf{x})$ \citep{de2022riemannian}.

For the forward diffusion we use rescaled Brownian motion given by $\mathbf{f}(\mathbf{x}, t) = 0, g(t) = \sqrt{\frac{d}{dt} \sigma^2(t)}$ where $\sigma(t)$ is the noise scale. Specifically, we use an exponential diffusion $\sigma(t) = \sigma^{1-t}_\text{min}\sigma^t_\text{max}$ as in \citet{song2019generative}, with $\sigma_\text{min}=0.01\pi$, $\sigma_\text{max} = \pi, t \in (0, 1)$. Due to the compactness of the manifold, however, the prior $p_T(\mathbf{x})$ is no longer a Gaussian, but a \emph{uniform} distribution over $M$. 

Training the score model with denoising score matching requires a procedure to sample from the perturbation kernel $p_{t\mid 0}(\mathbf{x}' \mid \mathbf{x})$ of the forward diffusion and compute its score. We view the torus $\mathbb{T}^m \cong [0, 2\pi)^m$ as the quotient space $\mathbb{R}^m/2\pi\mathbb{Z}^m$ with equivalence relations $(\tau_1, \ldots \tau_m) \sim (\tau_1+2\pi, \ldots, \tau_m) \ldots \sim (\tau_1, \ldots \tau_m+2\pi)$. Hence, the perturbation kernel for rescaled Brownian motion on $\mathbb{T}^m$ is the \emph{wrapped normal distribution} on $\mathbb{R}^m$; that is, for any $\boldsymbol{\tau}, \boldsymbol{\tau}' \in [0, 2\pi)^m$, we have
\begin{equation} \label{eq:torus_score}
    p_{t\mid 0}(\boldsymbol{\tau}' \mid \boldsymbol{\tau}) \propto \sum_{\mathbf{d} \in \mathbb{Z}^m} \exp\left(-\frac{||\boldsymbol{\tau} - \boldsymbol{\tau}' + 2\pi\mathbf{d}||^2}{2\sigma^2(t)}\right)
\end{equation}
where $\sigma(t)$ is the noise scale of the perturbation kernel $p_{t\mid 0}$. We thus sample from the perturbation kernel by sampling from the corresponding unwrapped isotropic normal and taking elementwise $\mod 2\pi$. The scores of the kernel are pre-computed using a numerical approximation. During training, we sample times $t$ at uniform and minimize the denoising score matching loss
\begin{equation} \label{eq:dsm}
    J_\text{DSM}(\theta) = \mathbb{E}_t\left[\lambda(t)\mathbb{E}_{\boldsymbol{\tau}_0\sim p_0,\boldsymbol{\tau}_t\sim p_{t\mid 0}(\cdot \mid \boldsymbol{\tau}_0)}\left[||\mathbf{s}(\boldsymbol{\tau}_t, t) - \nabla_{\boldsymbol{\tau}_t} \log p_{t\mid 0}(\boldsymbol{\tau}_t \mid \boldsymbol{\tau}_0)||^2\right]\right]
\end{equation}
where the weight factors $
    \lambda(t) = 1/\mathbb{E}_{\boldsymbol{\tau} \sim p_{t\mid 0}(\cdot \mid 0)}\left[||\nabla_{\boldsymbol{\tau}} \log p_{t\mid 0}(\boldsymbol{\tau} \mid \mathbf{0})||^2\right]$
are also precomputed. As the tangent space $T_{\boldsymbol{\tau}}\mathbb{T}^m$ is just $\mathbb{R}^m$, all the operations in the loss computation are the familiar ones.

For inference, we first sample from a uniform prior over the torus. We then discretize and solve the reverse diffusion with a geodesic random walk; however, since the exponential map on the torus (viewed as a quotient space) is just $\exp_{\boldsymbol{\tau}}(\boldsymbol{\delta}) = \boldsymbol{\tau} + \boldsymbol{\delta} \mod 2\pi$, the geodesic random walk is equivalent to the wrapping of the random walk on $\mathbb{R}^m$.

\subsection{Torsional score framework} \label{sec:score_framework}

\begin{wrapfigure}[28]{r}{0.3\textwidth}
\vspace{-10pt}
    \begin{center}
    \includegraphics[width=0.3\textwidth]{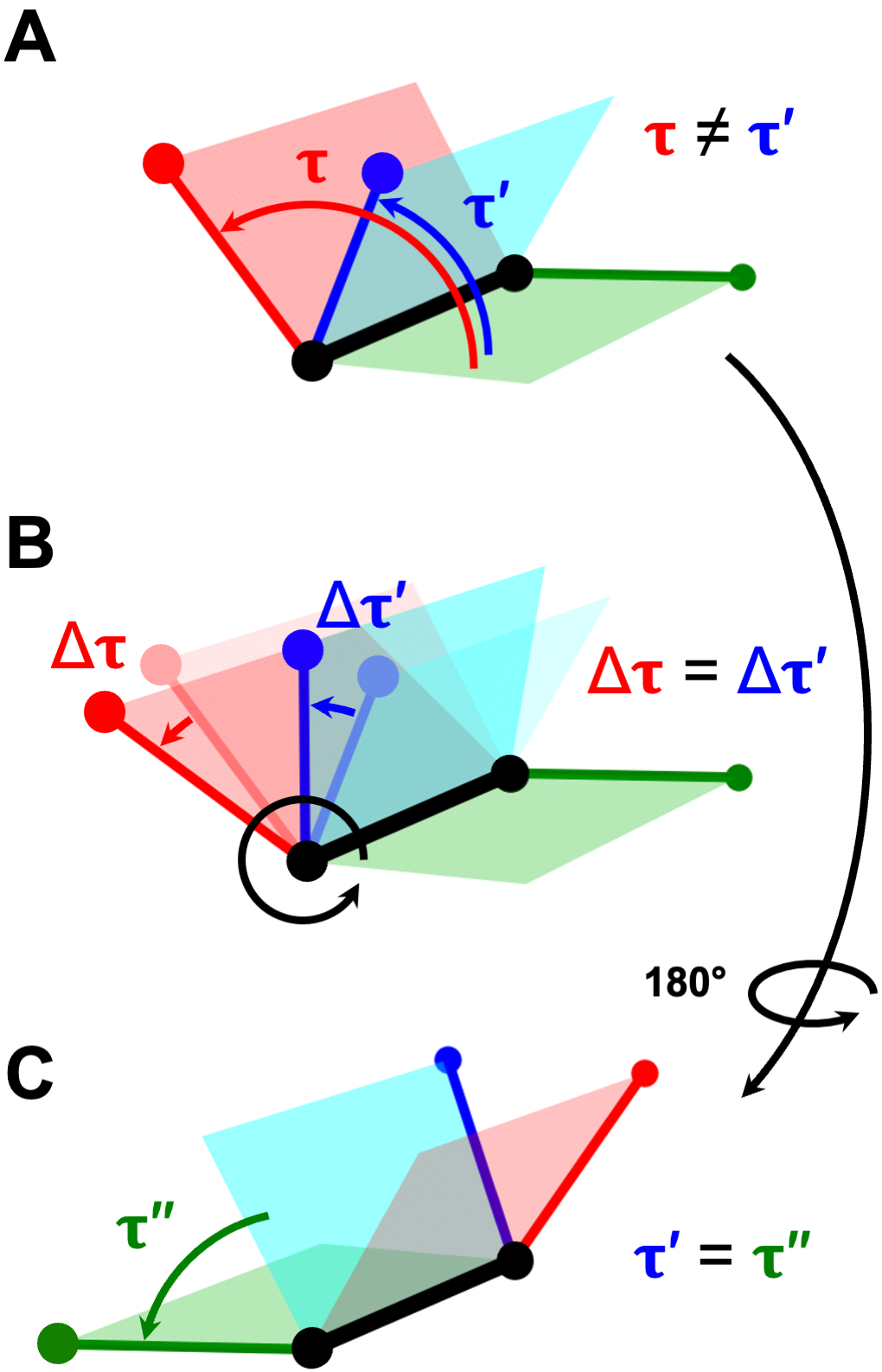}
    \vspace{-12pt}
    \caption{\textbf{A}: The torsion $\tau$ around a bond depends on a choice of neighbors. \textbf{B}: The \emph{change} $\Delta\tau$ caused by a relative rotation is the same for all choices. \textbf{C}: The sign of $\Delta\tau$ is unambiguous because given the same neighbors, $\tau$ does not depend on bond direction.}
    \label{fig:torsion}
    \end{center}
\end{wrapfigure}

While we have defined the diffusion process over intrinsic coordinates, learning a score model $\mathbf{s}(\boldsymbol{\tau}, t)$ directly over intrinsic coordinates is potentially problematic for several reasons. First, the dimensionality $m$ of the torsional space depends on the molecular graph $G$. Second, the mapping from torsional space to physically distinct conformers depends on $G$ and local structures $L$, but it is unclear how to best provide these to a model over $\mathbb{T}^m$. Third, there is no canonical choice of independent intrinsic coordinates $(L, \boldsymbol{\tau})$; in particular, the torsion angle at a rotatable bond can be defined as any of the dihedral angles at that bond, depending on an arbitrary choice of reference neighbors (Figure \ref{fig:torsion} and Appendix~\ref{app:def}). Thus, even with fixed $G$ and $L$, the mapping from $\mathbb{T}^m$ to conformers is ill-defined. This posed a significant challenge to prior works using intrinsic coordinates \citep{ganea2021geomol}.

To circumvent these difficulties, we instead consider a conformer $C\in \mathcal{C}_G$ in terms of its \emph{extrinsic} (or Cartesian) coordinates---that is, as a point cloud in 3D space, defined up to global roto-translation: $\mathcal{C}_G \cong \mathbb{R}^{3n} / SE(3)$. Then, we construct the score model $\mathbf{s}_G(C, t)$ as a function over $\mathcal{C}_G$ rather than $\mathbb{T}^m$. The outputs remain in the tangent space of $\mathbb{T}^m$, which is just $\mathbb{R}^m$. Such a score model is simply an $SE(3)$-\emph{invariant} model over point clouds in 3D space $\mathbf{s}_G: \mathbb{R}^{3n} \times [0, T] \mapsto \mathbb{R}^{m}$ conditioned on $G$. Thus, we have reduced the problem of learning a score on the torus, conditioned on the molecular graph and local structure, to the much more familiar problem of predicting $SE(3)$-invariant scalar quantities---one for each bond---from a 3D conformer.

It may appear that we still need to choose a definition of each torsion angle $\tau_i$ so that we can sample from  $p_{t\mid 0}(\cdot \mid \boldsymbol{\tau})$ during training and solve the reverse SDE over $\boldsymbol{\tau}$ during inference. However, we leverage the following insight: given \emph{fixed local structures}, the action on $C$ of changing a single torsion angle $\tau_i$ by some $\Delta \tau_i$ can be applied without choosing a definition (Figure \ref{fig:torsion}). Geometrically, this action is a (signed) relative rotation of the atoms on opposite sides of the bond and can be applied directly to the atomic coordinates in 3D. The geometric intuition can be stated as follows (proven in Appendix~\ref{app:propositions} and discussed further in Appendix~\ref{app:discuss:torsion}).
\begin{prop} \label{prop:torsion}
Let $(b_i, c_i)$ be a rotatable bond, let $\mathbf{x}_{\mathcal{V}(b_i)}$ be the positions of atoms on the $b_i$ side of the molecule, and let $R(\boldsymbol{\theta}, x_{c_i}) \in SE(3)$ be the rotation by Euler vector $\boldsymbol{\theta}$ about $x_{c_i}$. Then for $C, C' \in \mathcal{C}_G$, if $\tau_i$ is any definition of the torsion angle around bond $(b_i, c_i)$,

\begin{equation} 
    \begin{aligned}
        \tau_i(C') &= \tau_i(C) + \theta\\
        \tau_j(C') &= \tau_j(C) \quad \forall j\neq i
    \end{aligned}
    \qquad \text{if} \qquad
    \exists \mathbf{x} \in C, \mathbf{x'} \in C'\ldotp \quad
    \begin{aligned}
    \mathbf{x}'_{\mathcal{V}(b_i)} &=  \mathbf{x}_{\mathcal{V}(b_i)} \\
    \mathbf{x}'_{\mathcal{V}(c_i)} &=  R\left(\theta \, \mathbf{\hat r}_{b_ic_i}, x_{c_i} \right)\mathbf{x}_{\mathcal{V}(c_i)}
    \end{aligned}
\end{equation}
where $\mathbf{\hat r}_{b_ic_i} = (x_{c_i} - x_{b_i})/||x_{c_i}-x_{b_i}||$.
\end{prop}
To apply a torsion update $\Delta\boldsymbol{\tau} = (\Delta\tau_1,\ldots\Delta\tau_m)$ involving all bonds, we apply $\Delta\tau_i$ sequentially in any order. Then, since training and inference only make use of torsion updates $\Delta\boldsymbol{\tau}$, we work solely in terms of 3D point clouds and updates applied to them. To draw local structures $L$ from RDKit, we draw full 3D conformers $C \in \mathcal{C}_G$ and then randomize all torsion angles to sample uniformly over $\mathbb{T}^m$. To solve the reverse SDE, we repeatedly predict torsion updates directly from, and apply them directly to, the 3D point cloud. Therefore, since our method never requires a choice of reference neighbors for any $\tau_i$, it is manifestly invariant to such a choice. These procedures are detailed in Appendix \ref{app:summary_procedures}.

\subsection{Parity equivariance} \label{sec:parity}

The torsional score framework presented thus far requires an $SE(3)$-invariant model. However, an additional symmetry requirement arises from the fact that the underlying physical energy is invariant, or extremely nearly so, under \emph{parity inversion} \citep{quack2002important}. Thus our learned density should respect $p(C) = p(-C)$ where $-C = \{-\mathbf{x} \mid \mathbf{x} \in C\}$. In terms of the conditional distribution over torsion angles, we require $p(\boldsymbol{\tau}(C) \mid L(C)) = p(\boldsymbol{\tau}(-C) \mid L(-C))$. Then,
\begin{prop} \label{prop:parity}
    If $p(\boldsymbol{\tau}(C) \mid L(C)) = p(\boldsymbol{\tau}(-C) \mid L(-C))$, then for all diffusion times $t$,
    \begin{equation}
        \nabla_{\boldsymbol{\tau}} \log p_t(\boldsymbol{\tau}(C) \mid L(C)) = -\nabla_{\boldsymbol{\tau}} \log p_t(\boldsymbol{\tau}(-C) \mid L(-C)) 
    \end{equation}
\end{prop}

Because the score model seeks to learn $\mathbf{s}_G(C, t) = \nabla_{\boldsymbol{\tau}} \log p_t(\boldsymbol{\tau}(C) \mid L(C))$, we must have $\mathbf{s}_G(C, t) = -\mathbf{s}_G(-C, t)$. Thus, the score model must be \emph{invariant} under $SE(3)$ but \emph{equivariant} (change sign) under parity inversion of the input point cloud--- i.e. it must output a set of \emph{pseudoscalars} in $\mathbb{R}^m$.

\subsection{Score network architecture} \label{sec:score_model}

Based on sections \ref{sec:score_framework} and \ref{sec:parity}, the desiderata for the score model are:
\begin{center}
    \emph{Predict a pseudoscalar $\delta\tau_i := \partial \log p / \partial \tau_i \in\mathbb{R}$ that is $SE(3)$-invariant and parity equivariant\\for every rotatable bond in a 3D point cloud representation of a conformer.\\}
\end{center}

While there exist several GNN architectures which are $SE(3)$-equivariant \citep{jing2020learning, satorras2021n}, their $SE(3)$-invariant outputs are also parity invariant and, therefore, cannot satisfy the desired symmetry. Instead, we leverage the ability of equivariant networks based on tensor products \citep{thomas2018tensor, e3nn} to produce pseudoscalar outputs.

Our architecture, detailed in Appendix \ref{app:architecture}, consists of an embedding layer, a series of atomic convolution layers, and a final bond convolution layer. The first two closely follow the architecture of Tensor Field Networks \citep{thomas2018tensor}, and produce learned feature vectors for each atom. The final bond convolution layer constructs tensor product filters spatially centered on every rotatable bond and aggregates messages from neighboring atom features. We extract the pseudoscalar outputs of this filter to produce a single real-valued pseudoscalar prediction $\delta\tau_i$ for each rotatable bond. 

Naively, the bond convolution layer could be constructed the same way as the atomic convolution layers, i.e., with spherical harmonic filters. However, to supply information about the orientation of the bond about which the torsion occurs, we construct a filter from the product of the spherical harmonics with a representation of the bond (Figure \ref{fig:overview}D). Because the convolution conceptually resembles computing the torque, we call this final layer the \emph{pseudotorque} layer.

\subsection{Likelihood} \label{sec:likelihood}

By using the probability flow ODE, we can compute the likelihood of any sample $\boldsymbol{\tau}$ as follows \citep{song2021score,de2022riemannian}:
\begin{equation} \label{eq:likelihood}
    \log p_0(\boldsymbol{\tau}_0) = \log p_T(\boldsymbol{\tau}_T) - \frac{1}{2} \int_0^T g^2(t) \; \nabla_{\boldsymbol{\tau}} \cdot \mathbf{s}_G(\boldsymbol{\tau}_t, t) \; dt
\end{equation}
In \cite{song2021score}, the divergence term is approximated via Hutchinson's method \citep{hutchinson1989stochastic}, which gives an unbiased estimate of $\log p_0(\boldsymbol{\tau})$. However, this gives a \emph{biased} estimate of $ p_0(\boldsymbol{\tau})$, which is unsuitable for our applications. Thus, we compute the divergence term directly, which is feasible here (unlike in Euclidean diffusion) due to the reduced dimensionality of the torsional space.

The above likelihood is in \textit{torsional} space $p_G(\boldsymbol{\tau} \mid L), \boldsymbol{\tau} \in \mathbb{T}^m$, but to enable compatibility with the Boltzmann measure $e^{-E(\mathbf{x})/kT}$, it is desirable to interconvert this with a likelihood in \textit{Euclidean} space $p(\mathbf{x} \mid L), \mathbf{x} \in \mathbb{R}^{3n}$. A factor is necessary to convert between the volume element in torsional space and in Euclidean space (full derivation in Appendix~\ref{app:propositions}):
\begin{prop} \label{prop:euclidean}
Let $\mathbf{x} \in C(\boldsymbol{\tau}, L)$ be a centered\footnote{Additional formalism is needed for translations, but it is independent of the conformer and can be ignored.}
conformer in Euclidean space. Then,
\begin{equation}
    p_G(\mathbf{x} \mid L) = \frac{p_G(\boldsymbol{\tau} \mid L)}{ 8 \pi^2 \sqrt{\det g}}
    \quad \mathrm{where} \ \
    g_{\alpha\beta} = 
    \sum_{k=1}^{n} 
    J^{(k)}_{\alpha} \cdot J^{(k)}_{\beta}
\end{equation}
where the indices $\alpha,\beta$ are integers between 1 and $m+3$. For $1 \leq \alpha \leq m$, $J^{(k)}_\alpha$ is defined as
    \begin{align}
    \label{eqn:basisvec}
        J^{(k)}_{i} &= \tilde J^{(k)}_{i} - \frac 1 n \sum_{\ell=1}^{n} \tilde J^{(\ell)}_{i}
        \quad
        \mathrm{with} \ \
        \tilde J^{(\ell)}_{i} =
        \begin{cases}
            0 & \ell \in \mathcal{V}(b_i), \\
            \frac{\mathbf{x}_{b_i} - \mathbf{x}_{c_i}} {||\mathbf{x}_{b_i} - \mathbf{x}_{c_i}||}
            \times
            \left( \mathbf{x}_\ell - \mathbf{x}_{c_i} \right),
            & \ell \in \mathcal{V}(c_i),
        \end{cases}
    \end{align}
    and for $\alpha \in \{m+1, m+2, m+3\}$ as
    \begin{align}
    \label{eq:omegajacobian}
        J^{(k)}_{m+1} &= \mathbf{x}_k \times \hat{x},
        \qquad
        J^{(k)}_{m+2} = \mathbf{x}_k \times \hat{y},
        \qquad
        J^{(k)}_{m+3} = \mathbf{x}_k \times \hat{z},
        \qquad
    \end{align}
    where $(b_i, c_i)$ is the freely rotatable bond for torsion angle $i$, $\mathcal{V}(b_i)$ is the set of all nodes on the same side of the bond as $b_i$, and $\hat x, \hat y, \hat z$ are the unit vectors in the respective directions.
\end{prop}

\subsection{Energy-based training} \label{sec:energy}

By computing likelihoods, we can train torsional diffusion models to match the Boltzmann distribution over torsion angles using the energy function. At a high level, we minimize the usual score matching loss, but with simulated samples from the Boltzmann distribution rather than data samples. The procedure therefore consists of two stages: resampling and score matching, which are tightly coupled during training (Algorithm 1). In the \emph{resampling} stage, we use the model as an importance sampler for the Boltzmann distribution, where Proposition \ref{prop:euclidean} is used to compute the (unnormalized) torsional Boltzmann density $\tilde{p}_G(\boldsymbol{\tau}\mid L)$. In the \emph{score-matching} stage, the importance weights are used to approximate the denoising score-matching loss with expectations taken over $\tilde{p}_G(\boldsymbol{\tau}\mid L)$. As the model learns the score, it improves as an importance sampler. 

\begin{wrapfigure}[11]{r}{0.515\textwidth}
\vspace{-12pt}
\begin{small}
\begin{algorithm}[H]
\caption{Energy-based training epoch}\label{alg:energy}
\KwIn{Boltzmann density $\tilde{p}$, training pairs $\{(G_i, L_i)\}_i$, torsional diffusion model $q$}
\For{\textbf{each} $(G_i, L_i)$}{
    Sample $\boldsymbol{\tau}_1, \ldots \boldsymbol{\tau}_K \sim q_{G_i}(\boldsymbol{\tau}\mid L_i)$\;
    \For{$k\leftarrow 1$ \KwTo $K$}{
        $\tilde{w}_k = \tilde{p}_{G_i}(\boldsymbol{\tau}_k\mid L_i)/q_{G_i}(\boldsymbol{\tau}_k\mid L_i)$\;
    }
    Approximate $J_\text{DSM}$ for $p_0 \propto \tilde{p}$ using $\{(\tilde{w}_i, \boldsymbol{\tau}_i)\}_i$\;
    Minimize $J_\text{DSM}$\;
}
\end{algorithm}
\end{small}
\end{wrapfigure}

This training procedure differs substantially from that of existing Boltzmann generators, which are trained as flows with a loss that directly depends on the model density. In contrast, we \emph{train} the model as a score-based model, but \emph{use} it as a flow---both during training and inference---to generate samples. The model density is needed only to reweight the samples to approximate the target density. Since in principle the model used for resampling does not need to be the same as the model being trained,\footnote{For example, if the resampler were perfect, the procedure would reduce to normal denoising score matching.} we can use very few steps (a shallow flow) during resampling to accelerate training, and then increase the number of steps (a deeper flow) for better approximations during inference---an option unavailable to existing Boltzmann generators.

\section{Experiments} \label{sec:experiments}

We evaluate torsional diffusion by comparing the generated and ground-truth conformers in terms of ensemble RMSD (Section \ref{sec:ensemble_quality}) and properties (Section \ref{sec:ensemble_prop}). Section \ref{sec:conf_matching} first discusses a preprocessing procedure required to train a conditional model $p_G(\boldsymbol{\tau}\mid L)$. Section \ref{sec:boltzmann} concludes with torsional Boltzmann generators. See Appendix~\ref{app:results} for additional results, including ablation experiments.

\subsection{Conformer matching} \label{sec:conf_matching}
In focusing on $p_G(\boldsymbol{\tau}\mid L)$, we have assumed that we can sample local structures $L \sim p_G(L)$ with RDKit. While this assumption is very good in terms of RMSD, the RDKit marginal $\hat p_G(L)$ is only an approximation of the ground truth $p_G(L)$. Thus, if we train on the denoising score-matching loss with ground truth conformers---i.e., conditioned on ground truth local structures---there will be a distributional shift at test time, where only approximate local structures from $\hat p_G(L)$ are available. We found that this shift significantly hurts performance.

We thus introduce a preprocessing procedure called \textit{conformer matching}. In brief, for the \emph{training} split only, we substitute each ground truth conformer $C$ with a synthetic conformer $\hat{C}$ with local structures $\hat{L} \sim \hat p_G(L)$ and made as similar as possible to $C$. That is, we use RDKit to generate $\hat{L}$ and change torsion angles $\hat{\boldsymbol{\tau}}$ to minimize $\rmsd(C, \hat{C})$. Naively, we could sample $\hat{L} \sim \hat p_G(L)$ independently for each conformer, but this eliminates any possible dependence between $L$ and $\boldsymbol{\tau}$ that could serve as training signal. Instead, we view the distributional shift as a domain adaptation problem that can be solved by optimally aligning $p_G(L)$ and $\hat p_G(L)$. See Appendix \ref{app:matching} for details.

\subsection{Experimental setup} \label{sec:exp_setup}

\textbf{Dataset}\quad We evaluate on the GEOM dataset \citep{axelrod2020geom}, which provides gold-standard conformer ensembles generated with metadynamics in CREST \citep{pracht2020automated}. We focus on GEOM-DRUGS---the largest and most pharmaceutically relevant part of the dataset---consisting of 304k drug-like molecules (average 44 atoms). To test the capacity to extrapolate to the largest molecules, we also collect from GEOM-MoleculeNet all species with more than 100 atoms into a dataset we call GEOM-XL and use it to evaluate models trained on DRUGS. Finally, we train and evaluate models on GEOM-QM9, a more established dataset but with significantly smaller molecules (average 11 atoms). Results for GEOM-XL and GEOM-QM9 are in Appendix \ref{app:results}.

\textbf{Evaluation}\quad We use the train/val/test splits from \cite{ganea2021geomol} and use the same metrics to compare the generated and ground truth conformer ensembles: Average Minimum RMSD (AMR) and Coverage. These metrics are reported both for Recall (R)---which measures how well the generated ensemble covers the ground-truth ensemble---and Precision (P)---which measures the accuracy of the generated conformers. See Appendix \ref{app:exp_details} for exact definitions and further details. Following the literature, we generate $2K$ conformers for a molecule with $K$ ground truth conformers.

\textbf{Baselines}\quad We compare with the strongest existing methods from Section \ref{sec:background}. Among cheminformatics methods, we evaluate RDKit ETKDG \citep{riniker2015better}, the most established open-source package, and OMEGA \citep{hawkins2010conformer, hawkins2012conformer}, a commercial software in continuous development. Among machine learning methods, we evaluate GeoMol \citep{ganea2021geomol} and GeoDiff \citep{xu2021geodiff}, which have outperformed all previous models on the evaluation metrics. Note that GeoDiff originally used a small subset of the DRUGS dataset, so we retrained it using the splits from \cite{ganea2021geomol}.

\begin{table}[t]
\caption{Quality of generated conformer ensembles for the GEOM-DRUGS test set in terms of Coverage (\%) and Average Minimum RMSD (\AA). We compute Coverage with a threshold of $\delta=0.75$~\AA\ to better distinguish top methods. Note that this is different from most prior works, which used $\delta=1.25$ \AA.}\label{tab:quality}
\centering
\begin{tabular}{l|cccc|cccc} \toprule
                & \multicolumn{4}{c|}{Recall} & \multicolumn{4}{c}{Precision}  \\
                  & \multicolumn{2}{c}{Coverage $\uparrow$} & \multicolumn{2}{c|}{AMR $\downarrow$} & \multicolumn{2}{c}{Coverage $\uparrow$} & \multicolumn{2}{c}{AMR $\downarrow$} \\
Method & Mean & Med & Mean & Med & Mean & Med & Mean & Med \\ \midrule
RDKit ETKDG & 38.4 & 28.6 & 1.058 & 1.002 & 40.9 & 30.8 & 0.995 & 0.895 \\
OMEGA & 53.4 & 54.6 & 0.841 & 0.762 & 40.5 & 33.3 & 0.946 & 0.854 \\

GeoMol & 44.6 & 41.4 & 0.875 & 0.834 & 43.0 & 36.4 & 0.928 & 0.841 \\
GeoDiff & 42.1 & 37.8 & 0.835 & 0.809 & 24.9 & 14.5 & 1.136 & 1.090 \\ 
Torsional Diffusion & \textbf{72.7} & \textbf{80.0} & \textbf{0.582} & \textbf{0.565} & \textbf{55.2} & \textbf{56.9} & \textbf{0.778} & \textbf{0.729}     \\ \bottomrule
\end{tabular}
\end{table}

\subsection{Ensemble RMSD} \label{sec:ensemble_quality}

Torsional diffusion significantly outperforms all previous methods on GEOM-DRUGS (Table \ref{tab:quality} and Figure \ref{fig:coverage}), reducing by 30\% the average minimum recall RMSD and by 16\% the precision RMSD relative to the previous state-of-the-art method. Torsional diffusion is also the first ML method to consistently generate better ensembles than OMEGA. As OMEGA is a well-established product used in industry, this represents an essential step towards establishing the utility of conformer generation with machine learning. 

Torsional diffusion offers specific advantages over both GeoDiff and GeoMol, the most advanced prior machine learning methods. GeoDiff, a Euclidean diffusion model, requires 5000 denoising steps to obtain the results shown, whereas our model---thanks to the reduced degrees of freedom---requires only 20 steps. In fact, our model outperforms GeoDiff with as few as 5 denoising steps. As seen in Table \ref{tab:runtime}, this translates to enormous runtime improvements. 

Compared to torsional diffusion, GeoMol similarly makes use of intrinsic coordinates. However, since GeoMol can only access the molecular graph, it is less suited for reasoning about relationships that emerge only in a spatial embedding, especially between regions of the molecule that are distant on the graph. Our extrinsic-to-intrinsic score framework---which gives direct access to spatial relationships---addresses precisely this issue. The empirical advantages are most evident for the large molecules in GEOM-XL, on which GeoMol fails to improve consistently over RDKit (Appendix \ref{app:results}). On the other hand, because GeoMol requires only a single-forward pass, it retains the advantage of faster runtime compared to diffusion-based methods.

\begin{figure}[t]
    \caption{Mean coverage for recall (\emph{left}) and precision (\emph{right}) when varying the threshold value $\delta$ on GEOM-DRUGS.}\label{fig:coverage}
    \centering
    \includegraphics[width=0.495\textwidth]{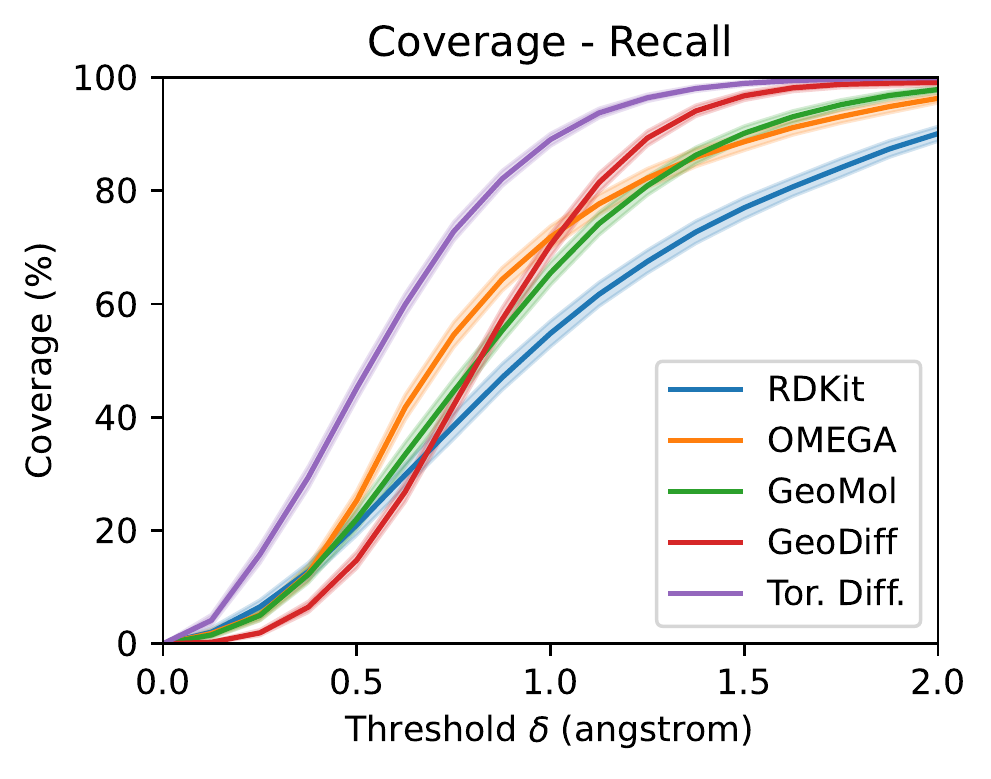}
    \includegraphics[width=0.495\textwidth]{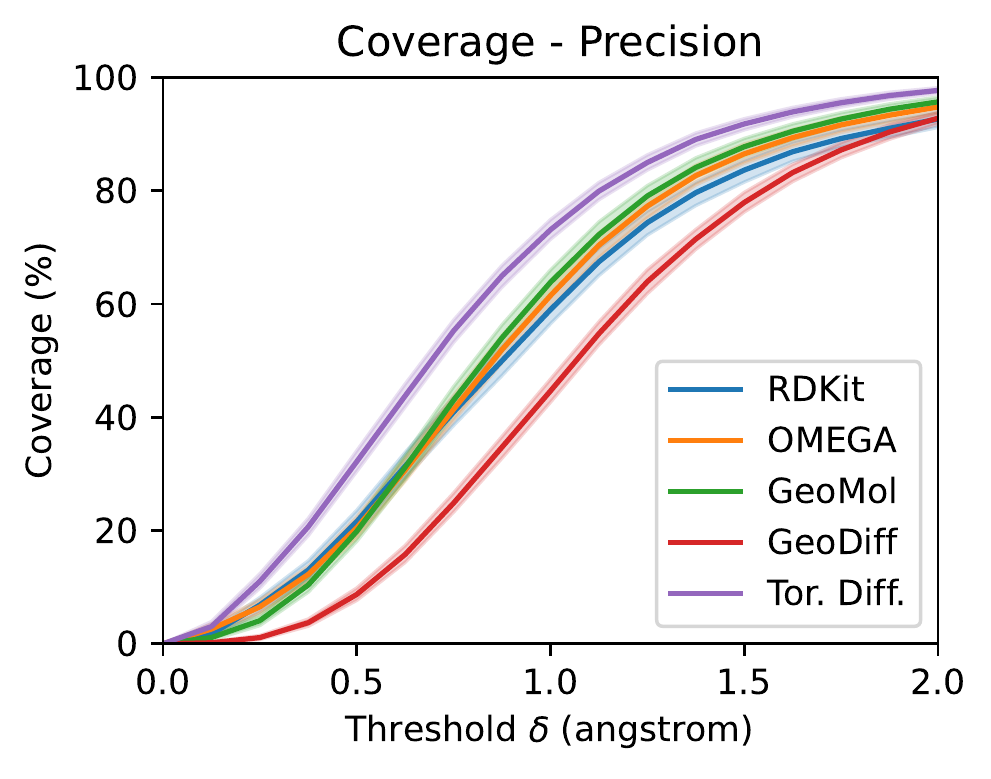}
\end{figure}

\subsection{Ensemble properties} \label{sec:ensemble_prop}

While RMSD gives a \emph{geometric} way to evaluate ensemble quality, we also consider the \emph{chemical} similarity between generated and ground truth ensembles. For a random 100-molecule subset of DRUGS, we generate $\min(2K, 32)$ conformers per molecule, relax the conformers with GFN2-xTB \citep{bannwarth2019gfn2},\footnote{Results without relaxation (which are less chemically meaningful) are in Appendix \ref{app:results}.} and compare the Boltzmann-weighted properties of the generated and ground truth ensembles. Specifically, the following properties are computed with xTB \citep{bannwarth2019gfn2}: energy $E$, dipole moment $\mu$, HOMO-LUMO gap $\Delta \epsilon$, and the minimum energy $E_{\min}$. The median errors for torsional diffusion and the baselines are shown in Table \ref{tab:boltzmann}. Our method produces the most chemically accurate ensembles, especially in terms of energy. In particular, we significantly improve over GeoMol and GeoDiff in finding the lowest-energy conformers that are only (on median) 0.13 kcal/mol higher in energy than the global minimum.

\begin{table}[t]
\parbox{.55\linewidth}{
    \caption{Median AMR and runtime (core-secs per conformer) of machine learning methods, evaluated on CPU for comparison with RDKit.} \label{tab:runtime}
    \centering
    \begin{tabular}{lcccc}\toprule  
    Method & Steps &  AMR-R & AMR-P & Runtime \\\midrule
    RDKit & -  & 1.002 & 0.895 &    \textbf{0.10}       \\
    GeoMol & - & 0.834 & 0.841 & 0.18 \\
    GeoDiff & 5000 & 0.809 & 1.090 & 305           \\ \midrule
    \multirow{3}{*}{\makecell{Torsional\\Diffusion}} 
    & 5       & 0.685 & 0.963 & 1.76          \\
    & 10      & 0.580 & 0.791 & 2.82          \\
    & 20      & \textbf{0.565} & \textbf{0.729} & 4.90          \\ \bottomrule
    \end{tabular}
}\hfill
\parbox{.43\linewidth}{
    \caption{Median absolute error of generated v.s. ground truth ensemble properties. $E, \Delta\epsilon, E_{\min}$ in kcal/mol, $\mu$ in debye.}\label{tab:properties}
    \centering
    \begin{tabular}{lcccc}
    \toprule  
    Method & $E$ & $\mu$ & $\Delta \epsilon$ & $E_{\min}$ \\\midrule
    RDKit & 0.81 & 0.52 & 0.75 & 1.16 \\
    OMEGA & 0.68 & 0.66 & 0.68 & 0.69 \\
    GeoMol & 0.42 & \textbf{0.34} & 0.59 & 0.40 \\
    GeoDiff & 0.31 & {0.35} & 0.89 & 0.39 \\
    Tor. Diff. & \textbf{0.22} & {0.35} & \textbf{0.54} & \textbf{0.13} \\ \bottomrule
    \end{tabular}
}
\end{table}

\clearpage

\subsection{Torsional Boltzmann generator} \label{sec:boltzmann}

\begin{wraptable}[14]{r}{6.2cm}
\vspace{-16pt}
\caption{Effective sample size (out of 32) given by importance sampling weights over the torsional Boltzmann density.}\label{tab:boltzmann}
\begin{tabular}{lcccc} \toprule
& & \multicolumn{3}{c}{Temp. (K)} \\ \cmidrule(lr){3-5}
Method & Steps & 1000 & 500 & 300 \\ \midrule
Uniform & -- & 1.71 & 1.21 & 1.02 \\ \midrule
\multirow{3}{*}{AIS} 
& 5 & 2.20 & 1.36 & 1.18 \\
& 20 & 3.12 & 1.76 & 1.30  \\ 
& 100 & 6.72 & 3.12 & 2.06 \\ \midrule
\multirow{2}{*}{\makecell[l]{Torsional\\BG}} 
& 5 & 7.28 & 3.60 & 3.04 \\
& 20 & \textbf{11.42} & \textbf{6.42} & \textbf{4.68} \\ \bottomrule
\end{tabular}
\end{wraptable}

Finally, we evaluate how well a torsional Boltzmann generator trained with MMFF \citep{halgren1996merck} energies can sample the corresponding Boltzmann density over torsion angles. We train and test on GEOM-DRUGS molecules with 3--7 rotatable bonds and use the local structures of the first ground-truth conformers. For the baselines, we implement annealed importance samplers (AIS) \citep{neal2001annealed} with Metropolis-Hastings steps over the torsional space and tune the variance of the transition kernels.

Table \ref{tab:boltzmann} shows the quality of the samplers in terms of the \emph{effective sample size} (ESS) given by the weights of 32 samples for each test molecule, which measures the $\alpha$-divergence (with $\alpha=2$) between the model and Boltzmann distributions \citep{midgley2021bootstrap}. Our method significantly outperforms the AIS baseline, and improves with increased step size despite being trained with only a 5-step resampler. Note that, since these evaluations are done on \emph{unseen} molecules, they are beyond the capabilities of existing Boltzmann generators.

\section{Conclusion} \label{sec:conclusion}

We presented \emph{torsional diffusion}, a method for generating molecular conformers based on a diffusion process restricted to the most flexible degrees of freedom. Torsional diffusion is the first machine learning model to significantly outperform standard cheminformatics methods and is orders of magnitude faster than previous Euclidean diffusion models. Using the exact likelihoods provided by our model, we also train the first system-agnostic Boltzmann generator.

There are several exciting avenues for future work. A natural extension is to relax the rigid local structure assumption by developing an efficient diffusion-based model over the full space of intrinsic coordinates while still incorporating chemical constraints. Moreover, torsional diffusion---or similar ideas---could be applicable to larger molecular systems, for which fast, parsimonious models of structural flexibility could benefit applications such as drug discovery and protein design.

\section*{Acknowledgments}
We pay tribute to Octavian-Eugen Ganea (1987-2022), dear colleague, mentor, and friend without whom this work would have never been possible. 

We thank Hannes St\"ark, Wenxian Shi, Xiang Fu, Felix Faltings, Jason Yim, Adam Fisch, Alex Wu, Jeremy Wohlwend, Peter Mikhael, and Saro Passaro for helpful feedback and discussions. We thank Lagnajit Pattanaik, Minkai Xu, and Simon Axelrod for their advice and support when working with, respectively, GeoMol, GeoDiff and the GEOM dataset. This work was supported by the Machine Learning for Pharmaceutical Discovery and Synthesis (MLPDS) consortium, the Abdul Latif Jameel Clinic for Machine Learning in Health, the DTRA Discovery of Medical Countermeasures Against New and Emerging (DOMANE) threats program, the DARPA Accelerated Molecular Discovery program and the Sanofi Computational Antibody Design grant. We acknowledge support from the Department of Energy Computational Science Graduate Fellowship (BJ), the Robert Shillman Fellowship (GC), and the NSF Graduate Research Fellowship (JC).

\newpage

\bibliography{references}
\bibliographystyle{plainnat}


\newpage

\appendix

\section{Definitions}\label{app:def}

Consider a molecular graph $G = (\mathcal{V}, \mathcal{E})$ and its space of possible conformers $\mathcal{C}_G$. A conformer is an assignment $\mathcal{V} \mapsto \mathbb{R}^3$ of each atom to a point in 3D-space, defined up to global rototranslation. For notational convenience, we suppose there is an ordering of nodes such that we can regard a mapping as a vector in $\mathbb{R}^{3n}$ where $n=|\mathcal{V}|$. Then a conformer $C \in \mathcal{C}_G$ is a set of $SE(3)$-equivalent vectors in $\mathbb{R}^{3n}$---that is, $\mathcal{C}_G \cong \mathbb{R}^{3n}/SE(3)$. This defines the space of conformers in terms of \emph{extrinsic} (or Cartesian) coordinates.

An \emph{intrinsic} (or internal) coordinate is a function over $\mathcal{C}_G$---i.e., it is an $SE(3)$-invariant function over $\mathbb{R}^{3n}$. There are four types of such coordinates typically considered:

\textbf{Bond lengths}. For $(a, b) \in \mathcal{E}$, the bond length $l_{ab} \in [0, \infty)$ is defined as $|x_{a} - x_b|$.

\textbf{Bond angles}. For $a, b, c \in \mathcal{V}$ such that $a, c \in \mathcal{N}(b)$, the bond angle $\alpha_{abc} \in [0, \pi]$ is defined by
\begin{equation}
    \cos \alpha_{abc} := \frac{(x_c - x_b)\cdot(x_a - x_b)}{|x_c - x_b||x_a - x_b|}
\end{equation}

\textbf{Chirality}. For $a \in \mathcal{V}$ with 4 neighbors $b, c, d, e \in \mathcal{N}(a)$, the chirality ${z}_{abcd}\in \{-1, 1\}$ is defined as
\begin{align}
    {z}_{abcde} := \sign\det\begin{pmatrix} 1 & 1 & 1 & 1 \\ x_b - x_a & x_c -x_a & x_d-x_a & x_e - x_a\end{pmatrix}
\end{align}
Similar quantities are defined for atoms with other numbers of neighbors. Chirality is often considered part of the specification of the molecule, rather than the conformer. See Appendix~\ref{app:discuss:isomerism} for additional discussion on this point.

\textbf{Torsion angles}. For $(b, c) \in \mathcal{E}$, with a choice of reference neighbors $a\in\mathcal{N}(b)\setminus\{c\}, d\in\mathcal{N}(c)\setminus\{b\}$, the torsion angle $\tau_{abcd} \in [0, 2\pi)$ is defined as the dihedral angle between planes $abc$ and $bcd$:
\begin{equation}\label{eq:torsion_def}
\begin{aligned}
    \cos \tau_{abcd} &= \frac{\mathbf{n}_{abc} \cdot \mathbf{n}_{bcd}}{|\mathbf{n}_{abc}||\mathbf{n}_{bcd}|}\\
    \sin \tau_{abcd} &= \frac{\mathbf{u}_{bc} \cdot (\mathbf{n}_{abc} \times \mathbf{n}_{bcd})}{|\mathbf{u}_{bc}||\mathbf{n}_{abc}||\mathbf{n}_{bcd}|}
\end{aligned}
\end{equation}
    
where $\mathbf{u}_{ab} = x_b - x_a$ and $\mathbf{n}_{abc}$ is the normal vector $\mathbf{u}_{ab}\times\mathbf{u}_{bc}$. Note that $\tau_{abcd} = \tau_{dcba}$---i.e., the dihedral angle is the same for four consecutively bonded atoms regardless of the direction in which they are considered.

A \textbf{complete set of intrinsic coordinates} of the molecule is a set of such functions $(f_1, f_2, \ldots)$ such that $F(C) = (f_1(C), f_2(C), \ldots)$ is a bijection. In other words, they fully specify a unique element of $\mathcal{C}_G$ without overparameterizing the space. In general there exist many possible such sets for a given molecular graph. We will not discuss further how to find such sets, as our work focuses on manipulating molecules in a way that holds fixed all $l, \alpha, z$ and only modifies (a subset of) torsion angles $\tau$.

As presently stated, the \textbf{torsion angle about a bond} $(b, c)\in\mathcal{E}$ is ill-defined, as it could be any $\tau_{abcd}$ with $a\in\mathcal{N}(b)\setminus\{c\}, d\in\mathcal{N}(c)\setminus\{b\}$. However, any complete set of intrinsic coordinates needs to only have at most one such $\tau_{abcd}$ for each bond $(b, c)$ \citep{ganea2021geomol}. Thus, we often refer to \emph{the} torsion angle about a bond $(b_i, c_i)$ as $\tau_i$ when reference neighbors $a_i, b_i$ are not explicitly stated.

\section{Propositions} \label{app:propositions}
\subsection{Torsion update}

Given a freely rotatable bond $(b_i, c_i)$, by definition removing $(b_i, c_i)$ creates two connected components $\mathcal{V}(b_i), \mathcal{V}(c_i)$. Then, consider torsion angle $\tau_j$ at a different bond $(b_j, c_j)$ with neighbor choices $a_j \in \mathcal{N}(b_j), d_j \in \mathcal{N}(c_j), a_j \neq c_j, d_j \neq b_j$. Without loss of generality, there are two cases
\begin{itemize}
    \item Case 1: $a_j, b_j, c_j, d_j \in \mathcal{V}(b_i)$
    \item Case 2: $d_j \in \mathcal{V}(c_i)$ and $a_j, b_j, c_j \in \mathcal{V}(b_i)$
\end{itemize} 
Note that in Case 2, $c_j = b_i$ and $d_j = c_i$ must hold because there is only one edge between $\mathcal{V}(b_i), \mathcal{V}(c_i)$. With these preliminaries we now restate the proposition:

\begin{proposition} 
Let $(b_i, c_i)$ be a rotatable bond, let $\mathbf{x}_{\mathcal{V}(b_i)}$ be the positions of atoms on the $b_i$ side of the molecule, and let $R(\boldsymbol{\theta}, x_{c_i}) \in SE(3)$ be the rotation by Euler vector $\boldsymbol{\theta}$ about $x_{c_i}$. Then for $C, C' \in \mathcal{C}_G$, if $\tau_i$ is any definition of the torsion angle around bond $(b_i, c_i)$,

\begin{equation} 
    \begin{aligned}
        \tau_i(C') &= \tau_i(C) + \theta\\
        \tau_j(C') &= \tau_j(C) \quad \forall j\neq i
    \end{aligned}
    \qquad \text{if} \qquad
    \exists \mathbf{x} \in C, \mathbf{x'} \in C'\ldotp \quad
    \begin{aligned}
    \mathbf{x}'_{\mathcal{V}(b_i)} &=  \mathbf{x}_{\mathcal{V}(b_i)} \\
    \mathbf{x}'_{\mathcal{V}(c_i)} &=  R\left(\theta \, \mathbf{\hat r}_{b_ic_i}, x_{c_i} \right)\mathbf{x}_{\mathcal{V}(c_i)}
    \end{aligned}
\end{equation}
where $\mathbf{\hat r}_{b_ic_i} = (x_{c_i} - x_{b_i})/||x_{c_i}-x_{b_i}||$.

\end{proposition}
\begin{small}
\begin{proof}
First we show $\tau_i(C') = \tau_i(C)+\theta$, for which it suffices to show $\tau_i(\mathbf{x}')=\tau_i(\mathbf{x})+\theta$. Because $a_i, b_i \in \mathcal{V}(b_i)$, $x_{a_i}' = x_{a_i}$ and $x_{b_i}' = x_{b_i}$. Since the rotation of $\mathbf{x}_{\mathcal{V}(c_i)}$ is centered at $x_{c_i}$, we have $x_{c_i}' = x_{c_i}$ as well. Now we consider $d_i$ and $\mathbf{u}'_{cd} = x'_{d_i} - x_{c_i}$. By the Rodrigues rotation formula,
\begin{equation}
\begin{aligned}
    \mathbf{u}'_{cd} &= \mathbf{u}_{cd}\cos\theta + \frac{\mathbf{n}_{bcd}}{|\mathbf{u}_{bc}|}\sin\theta + \frac{\mathbf{u}_{bc}}{|\mathbf{u}_{bc}|}\left(\frac{\mathbf{u}_{bc}}{|\mathbf{u}_{bc}|}\cdot \mathbf{u}_{cd} \right)(1-\cos\theta)
\end{aligned}
\end{equation}
Then we have
\begin{equation}
\begin{aligned}
    \mathbf{n}'_{bcd} = \mathbf{u}_{bc} \times \mathbf{u}'_{cd} &=\mathbf{n}_{bcd}\cos\theta - \left(\mathbf{n}_{bcd} \times \frac{\mathbf{u}_{bc}}{|\mathbf{u}_{bc}|}\right)\sin\theta
\end{aligned}\end{equation}
To obtain $|\mathbf{n}'_{bcd}|$, note that since $\mathbf{n}_{bcd}\perp \mathbf{u}_{bc}$, 
\begin{equation}
    \bigg|\mathbf{n}_{bcd} \times \frac{\mathbf{u}_{bc}}{|\mathbf{u}_{bc}|}\bigg| = |\mathbf{n}_{bcd}|
\end{equation}
which gives $|\mathbf{n}'_{bcd}| = |\mathbf{n}_{bcd}|$.
Thus,
\begin{equation}
\begin{aligned}    
    \cos\tau'_i &= \frac{\mathbf{n}_{abc} \cdot \mathbf{n}'_{bcd}}{|\mathbf{n}_{abc}||\mathbf{n}_{bcd}|} = \frac{\mathbf{n}_{abc} \cdot \mathbf{n}_{bcd}}{|\mathbf{n}_{abc}||\mathbf{n}_{bcd}|} \cos\theta - \frac{\mathbf{n}_{abc} \cdot \left(\mathbf{n}_{bcd}\times\mathbf{u}_{bc}\right)}{|\mathbf{n}_{abc}||\mathbf{n}_{bcd}||\mathbf{u}_{bc}|} \sin\theta \\ 
    &= \cos\tau_i\cos\theta - \sin\tau_i\sin\theta = \cos (\tau_i + \theta)
\end{aligned}
\end{equation}
Similarly, 
\begin{equation}
\begin{aligned}
    \sin\tau'_i &= \frac{\mathbf{u}_{bc} \cdot (\mathbf{n}_{abc} \times \mathbf{n}'_{bcd})}{|\mathbf{u}_{bc}||\mathbf{n}_{abc}||\mathbf{n}_{bcd}|} = \frac{\mathbf{u}_{bc} \cdot (\mathbf{n}_{abc} \times \mathbf{n}_{bcd})}{|\mathbf{u}_{bc}||\mathbf{n}_{abc}||\mathbf{n}_{bcd}|}\cos\theta - \frac{\mathbf{u}_{bc} \cdot (\mathbf{n}_{abc} \times \left(\mathbf{n}_{bcd}\times\mathbf{u}_{bc}\right))}{|\mathbf{u}_{bc}|^2|\mathbf{n}_{abc}||\mathbf{n}_{bcd}|}\sin\theta \\
    &= \sin\tau_i\cos\theta + \cos\tau_i\sin\theta = \sin(\tau_i + \theta)
\end{aligned}
\end{equation}
Therefore, $\tau'_i = \tau_i+\theta$

Now we show $\tau'_j = \tau_j$ for all $j\neq i$. Consider any such $j$. For Case 1, $x'_{a_j} = x_{a_j}, x'_{b_j} = x_{b_j}, x'_{c_j} = x_{c_j}, x'_{d_j} = x_{d_j}$ so clearly $\tau'_j = \tau_j$. For Case 2, $x'_{a_j} = x_{a_j}, x'_{b_j} = x_{b_j}, x'_{c_j} = x_{c_j}$ immediately. But because $d_j = c_i$, we also have $x'_{d_j} = x_{d_j}$. Thus,  $\tau'_j = \tau_j$.
\end{proof}
\end{small}
\subsection{Parity equivariance}
\begin{proposition}
    If $p(\boldsymbol{\tau}(C) \mid L(C)) = p(\boldsymbol{\tau}(-C) \mid L(-C))$, then for all diffusion times $t$,
    \begin{equation}
        \nabla_{\boldsymbol{\tau}} \log p_t(\boldsymbol{\tau}(C) \mid L(C)) = -\nabla_{\boldsymbol{\tau}} \log p_t(\boldsymbol{\tau}(-C) \mid L(-C)) 
    \end{equation}
\end{proposition}
\begin{small}

\begin{proof}
From Equation~\ref{eq:torsion_def} we see that for any torsion $\tau_i$, we have $\tau_i(-C) = -\tau_i(C)$; therefore $\boldsymbol{\tau}_i(-C) = -\boldsymbol{\tau}_i(C)$, which we denote $\boldsymbol{\tau}_-$. Also denote $\boldsymbol{\tau} := \boldsymbol{\tau}(C), p_t(\boldsymbol{\tau}) := p_t(\boldsymbol{\tau} \mid L(C))$ and $p'_t(\boldsymbol{\tau}_-) := p_t(\boldsymbol{\tau}_- \mid L(-C))$. We claim $p_t(\boldsymbol{\tau}) = p'_t(\boldsymbol{\tau}_-)$ for all $t$. Since the perturbation kernel (\eqref{eq:torus_score}) is parity invariant,
\begin{equation}
\begin{aligned}
    p'_t(\boldsymbol{\tau}_{-}) &= \int_{\mathbb{T}^m} p'_0(\boldsymbol{\tau}'_{-}) p_{t\mid 0}(\boldsymbol{\tau}_{-} \mid \boldsymbol{\tau}'_{-})\; d\boldsymbol{\tau}'_{-} \\
    &= \int_{\mathbb{T}^m} p_0(\boldsymbol{\tau}') p_{t\mid 0}(\boldsymbol{\tau} \mid \boldsymbol{\tau}')\; d\boldsymbol{\tau}'_{-} = p_t(\boldsymbol{\tau})
\end{aligned}
\end{equation}
Next, we have
\begin{equation}
\begin{aligned}
    \nabla_{\boldsymbol{\tau}}\log p'_t(\boldsymbol{\tau}_-) &= \frac{\partial\boldsymbol{\tau}_-}{\partial\boldsymbol{\tau}}\nabla_{\boldsymbol{\tau}^{-}}\log p'_t(\boldsymbol{\tau}_-) \\
    &= -\nabla_{\boldsymbol{\tau}}\log p_t(\boldsymbol{\tau}) 
\end{aligned}
\end{equation}
which concludes the proof.
\end{proof}
\end{small}

\subsection{Likelihood conversion}

\begin{proposition}
 \label{prop:euclidean_app}
Let $\mathbf{x} \in C(\boldsymbol{\tau}, L)$ be a centered conformer in Euclidean space. Then,
\begin{equation}
    p_G(\mathbf{x} \mid L) = \frac{p_G(\boldsymbol{\tau} \mid L)}{ 8 \pi^2 \sqrt{\det g}}
    \quad \mathrm{where} \ \
    g_{\alpha\beta} = 
    \sum_{k=1}^{n} 
    J^{(k)}_{\alpha} \cdot J^{(k)}_{\beta}
\end{equation}
where the indices $\alpha,\beta$ are integers between 1 and $m+3$. For $1 \leq \alpha \leq m$, $J^{(k)}_\alpha$ is defined as
    \begin{align}
    \label{eqn:basisvec_app}
        J^{(k)}_{i} &= \tilde J^{(k)}_{i} - \frac 1 n \sum_{\ell=1}^{n} \tilde J^{(\ell)}_{i}
        \quad
        \mathrm{with} \ \
        \tilde J^{(\ell)}_{i} =
        \begin{cases}
            0 & \ell \in \mathcal{V}(b_i), \\
            \frac{\mathbf{x}_{b_i} - \mathbf{x}_{c_i}} {||\mathbf{x}_{b_i} - \mathbf{x}_{c_i}||}
            \times
            \left( \mathbf{x}_\ell - \mathbf{x}_{c_i} \right),
            & \ell \in \mathcal{V}(c_i),
        \end{cases}
    \end{align}
    and for $\alpha \in \{m+1, m+2, m+3\}$ as
    \begin{align}
    \label{eq:omegajacobian_app}
        J^{(k)}_{m+1} &= \mathbf{x}_k \times \hat{x},
        \qquad
        J^{(k)}_{m+2} = \mathbf{x}_k \times \hat{y},
        \qquad
        J^{(k)}_{m+3} = \mathbf{x}_k \times \hat{z},
        \qquad
    \end{align}
    where $(b_i, c_i)$ is the freely rotatable bond for torsion angle $i$, $\mathcal{V}(b_i)$ is the set of all nodes on the same side of the bond as $b_i$, and $\hat x, \hat y, \hat z$ are the unit vectors in the respective directions.
\end{proposition}

\begin{proof}\begin{small}
Let $M$ be $(m+3)$-dimensional manifold embedded in $3n$-dimensional Euclidean space formed by the set of all centered conformers with fixed local structures but arbitrary torsion angles and orientation. A natural set of coordinates for $M$ is $q^\alpha = \{\tau_1, \tau_2, \ldots, \tau_m, \omega_x, \omega_y, \omega_z\}$, where $\tau_i$ is the torsion angle at bond $i$ and $\omega_x, \omega_y, \omega_z$ define the global rotation about the center of mass:
\begin{equation}
    \mathbf{x}_k = \tilde {\mathbf{x}}_k  - \frac 1 n \sum_{\ell=1}^n {\tilde{\mathbf{x}}}_\ell
    \quad
    \mathrm{where}
    \quad
    {\tilde{\mathbf{x}}}_\ell = e^{\Lambda(\omega)} \mathbf{x}^\prime_k,
    \quad
    \Lambda(\omega) =
    \begin{pmatrix}
        0 & -\omega_z & \omega_y \\
        \omega_z & 0 & -\omega_x \\
        -\omega_y & \omega_x & 0
    \end{pmatrix}.
\end{equation}
Here ${\mathbf{x}}_k^\prime$ is the position of atom $k$ as determined by the torsion angles, without centering or global rotations, and $\omega_x, \omega_y, \omega_z$ are rotation about the $x$, $y$, and $z$ axis respectively.

Consider the set of covariant basis vectors
\begin{equation}
    \mathbf{J}_\alpha = \frac{\partial \mathbf{x}}{\partial q^\alpha}.
\end{equation}
and corresponding the covariant components of the metric tensor,
\begin{equation}
    g_{\alpha \beta} = \mathbf{J}_\alpha \cdot \mathbf{J}_\beta = \frac{\partial \mathbf{x}}{\partial q^\alpha} \cdot \frac{\partial \mathbf{x}}{\partial q^\beta}.
\end{equation}
The conversion factor between torsional likelihood and Euclidean likelihood is given by
\begin{equation}
    \int \sqrt{\det \mathbf{g}} \, d^{3} \omega,
    \label{eq:conversionfactor_app}
\end{equation}
where $\sqrt{\det \mathbf{g}} \, d^{m+3}q$ is the invariant volume element on $M$ \citep{carroll2019spacetime}, and the integration over $\omega$ marginalizes over the uniform distribution over global rotations. The calculation of Eq.~\ref{eq:conversionfactor_app} proceeds as follows.

Let the position of the $k$'th atom be $x_k$, and let the three corresponding components of $\mathbf{J}_\alpha$ be $J^{(k)}_\alpha$. For $1 \leq i \leq m$, $J^{(k)}_i$is given by
\begin{equation}
   J^{(k)}_i = \frac{\partial}{\partial \tau_i}
   \left(
        \tilde {\mathbf{x}}_k
        - \frac 1 n \sum_{\ell=1}^n {\tilde{\mathbf{x}}}_\ell
   \right)
   = \tilde J^{(k)}_{i} - \frac 1 n \sum_{\ell=1}^{n} \tilde J^{(\ell)}_{i}
\end{equation}
where $\tilde J^{(k)}_{i} := \partial \tilde {\mathbf{x}}_{k}/\partial \tau_i$ is the displacement of atom $k$ upon an infinitesmal change in the torsion angle $\tau_i$, without considering the change in the center of mass. Clearly $\tilde J^{(b_i)}_{i} = \tilde J^{(c_i)}_{i} = 0$ because neither $b_i$ nor $c_i$ itself is displaced; furthermore, all atoms on the $b$ side of torsioning bond are not displaced, so $J^{(k)}_{i} = 0$ for all $k \in \mathcal{N}(b_i)$. The remaining atoms, in $\mathcal{N} (c_i)$, are rotated about the axis of the $(b_i,c_i)$ bond. The displacement per infinitesimal $\partial \tau_i$ is given by the cross product of the unit normal along the rotation axis, $(\tilde x_{c_i} - \tilde x_{b_i}) / {||\tilde x_{c_i} - \tilde x_{b_i}||}$, with the displacement from rotation axis, $\tilde x_k - \tilde x_{b_i}$. This cross product yields ${J}^{(k)}_\alpha$ in Eq.~\ref{eqn:basisvec_app}, where the tildes are dropped as relative positions do not depend on center of mass. For $\alpha \in \{m+1, m+2, m+3\}$, a similar consideration of the cross product with the rotation axis yields Eq.~\ref{eq:omegajacobian_app}. Finally, since none of the components of the metric tensor depend explicitly on $\omega$, the integration over $\omega$ in Eq.~\ref{eq:conversionfactor_app} is trivial and yields the volume over $SO(3)$ of $8 \pi^2$ \citep{chirikjian2011stochastic}, proving the proposition.
\end{small}\end{proof}

\section{Training and inference procedures} \label{app:summary_procedures}

Algorithms \ref{alg:training} and \ref{alg:inference} summarize, respectively, the training and inference procedures used for torsional diffusion. In practice, during training, we limit $K_G$ to 30 i.e. we only consider the first 30 conformers found by CREST (typically those with the largest Boltzmann weight). Moreover, molecules are batched and an Adam optimizer with a learning rate scheduler is used for optimization. For inference, to fairly compare with other methods from the literature, we follow \cite{ganea2021geomol} and set $K$ to be twice the number of conformers returned by CREST.

\begin{algorithm}[h]
\caption{Training procedure}\label{alg:training}
\KwIn{molecules $[G_0, ..., G_N]$ each with true conformers $[C_{G,1}, ... C_{G,K_G}]$, learning rate $\alpha$}
\KwOut{trained score model $\mathbf{s}_\theta$}
conformer matching process for each $G$ to get $[\hat{C}_{G,1}, ... \hat{C}_{G,K_G}]$\;
\For{epoch $\leftarrow 1$ \KwTo $\text{epoch}_{\max}$}{
    \For{$G$ \textbf{in} $[G_0, ..., G_N]$}{
        sample $t \in [0, 1]$ and $\hat{C} \in [\hat{C}_{G,1}, ... \hat{C}_{G,K_G}]$\;
        sample $\Delta \boldsymbol{\tau}$ from wrapped normal $p_{t\mid 0}(\cdot \mid \mathbf{0})$ with $\sigma = \sigma_{\min} ^{1-t} \, \sigma_{\max} ^t$\; 
        apply $\Delta \boldsymbol{\tau}$ to $\hat{C}$\;
        predict $\delta \boldsymbol{\tau} = \mathbf{s}_{\theta, G}(\hat{C}, t)$\;
        update $\theta \gets \theta - \alpha \nabla_{\theta} \lVert \delta \boldsymbol{\tau} - \nabla_{\Delta\boldsymbol{\tau}} p_{t\mid 0}(\Delta\boldsymbol{\tau} \mid \mathbf{0})  \rVert^2$\;
    }
}
\end{algorithm}

\begin{algorithm}[h]
\caption{Inference procedure}\label{alg:inference}
\KwIn{molecular graph $G$, number conformers $K$, number steps $N$}
\KwOut{predicted conformers  $[C_1, ... C_K]$}
generate local structures by obtaining conformers $[C_1, ... C_K]$ from RDKit\;
\For{$C$ \textbf{in} $[C_1, ... C_K]$}{
    sample $\Delta \boldsymbol{\tau} \sim U[0, 2\pi]^m$ and apply to $C$ to randomize torsion angles\;
    \For{n $\leftarrow N$ \KwTo $1$}{
        let $t = n/N, \; g(t) = \sigma_{\min} ^{1-t} \, \sigma_{\max} ^t\sqrt{2\ln(\sigma_{\max}/\sigma_{\min})}$\;
        predict $\delta \boldsymbol{\tau} = \mathbf{s}_{\theta, G}(\hat{C}, t)$\;
        draw $\mathbf{z}$ from wrapped normal with $\sigma^2 = 1/N$\;
        set $\Delta \boldsymbol{\tau} = (g^2(t)/N)\;\delta\boldsymbol{\tau} + g(t)\;\mathbf{z}$\;
        apply $\Delta\boldsymbol{\tau}$ to $C$\;
    }
}
\end{algorithm}

\section{Score network architecture} \label{app:architecture}

\paragraph{Overview} To perform the torsion score prediction under these symmetry constraints we design an architecture formed by three components: an embedding layer, a series of $K$ interaction layers and a pseudotorque layer. The pseudotorque layer produces pseudoscalar torsion scores $\delta\tau := \partial \log p/\partial\tau$ for every rotatable bond.
Following the notation from \cite{thomas2018tensor} and \cite{batzner2021se}, we represent the node representations as $V_{acm}^{(k, l, p)}$ a dictionary with keys the layer $k$, rotation order $l$ and parity $p$ that contains tensors with shapes $[ |\mathcal{V}|, n_l, 2l+1 ]$ corresponding to the indices of the node, channel and representation respectively.  We use the \verb|e3nn| library \citep{e3nn} to implement our architecture.

\begin{figure}[h]
    \centering
    \includegraphics[width=0.8\textwidth]{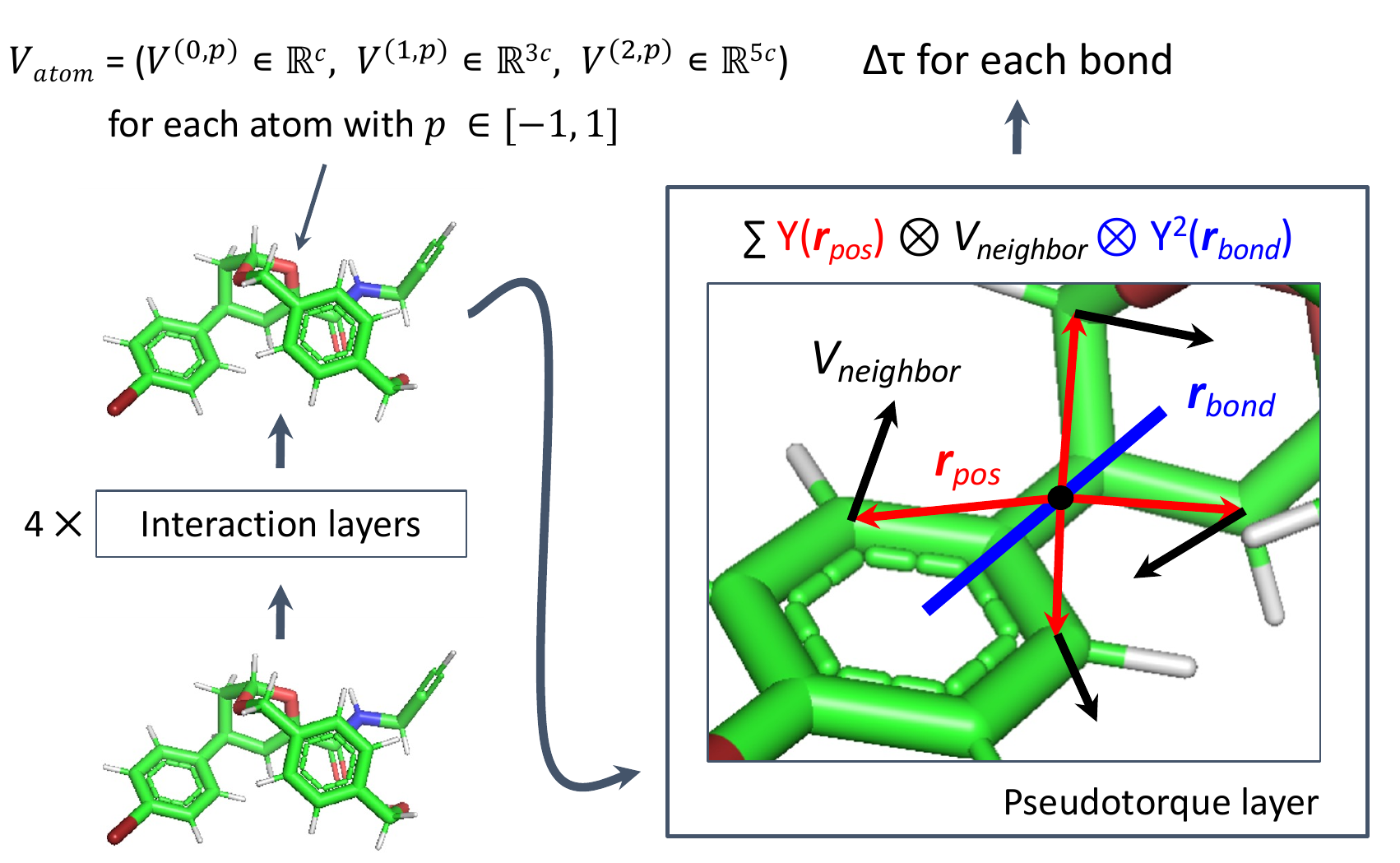}
    \caption{Overview of the architecture and visual intuition of the pseudotorque layer.}
    \label{fig:architecture}
\end{figure}

\paragraph{Embedding layer} In the embedding layer, we build a radius graph $(\mathcal{V}, \mathcal{E}_{r_{\max} })$ around each atom on top of the original molecular graph and generate initial scalar embeddings for nodes $V_a^{(0, 0, 1)}$ and edges $e_{ab}$ combining chemical properties, sinusoidal embeddings of time $\phi(t)$ \citep{vaswani2017attention} and, for the edges, a radial basis function representation of their length $\mu(r_{ab})$ \citep{schutt2017schnet}:
\begin{equation}
\begin{gathered}
\mathcal{E}_{r_{\max} } = \mathcal{E} \sqcup \{ (a,b) \mid r_{ab} < r_{\max} \}  \\
e_{ab} = \Upsilon^{(e)}(fe_{ab} || \mu(r_{ab}) || \phi(t)) \quad \forall (a,b) \in \mathcal{E}_{r_{\max}}  \\
V_a^{(0, 0, 1)} = \Upsilon^{(v)}(f_a || \phi(t))  \quad \forall a \in \mathcal{V}  
\end{gathered}
\end{equation}

where $\Upsilon^{(e)}$ and $\Upsilon^{(v)}$ are learnable two-layers MLPs, $r_{ab}$ is the Euclidean distance between atoms $a$ and $b$, $r_{\max} = 5\text{ \AA}$ is the distance cutoff, $f_a$ are the chemical features of atom $a$, $f_{ab}$ are the chemical features of bond $(a,b)$ if it was part of $\mathcal{E}$ and 0 otherwise.

The node and edge chemical features $f_a$ and $f_{ab}$ are constructed as in \citet{ganea2021geomol}. Briefly, the node features include atom identity, atomic number, aromaticity, degree, hybridization, implicit valence, formal charge, ring membership, and ring size, constituting a 74-dimensional vector for GEOM-DRUGS and 44-dimensional for QM9 (due to fewer atom types). The edge features are a 4 dimensional one-hot encoding of the bond type.

\paragraph{Interaction layers} The interaction layers are based on E(3)NN \citep{e3nn} convolutional layers. At each layer, for every pair of nodes in the graph, we construct messages using tensor products of the current irreducible representation of each node with the spherical harmonic representations of the normalized edge vector. These messages are themselves irreducible representations, which are weighted channel-wise by a scalar function of the current scalar representations of the two nodes and the edge and aggregated with Clebsch-Gordan coefficients.

At every layer $k$, for every node $a$, rotation order $l_o$, and output channel $c'$:
\begin{small}
\begin{equation}
\begin{gathered}
V_{ac'm_o}^{(k, l_o, p_o)} = \sum_{l_f, l_i, p_i} \sum_{m_f, m_i} C_{(l_i, m_i)(l_f, m_f)}^{(l_o,m_o)} \frac{1}{|\mathcal{N}_a|} \sum_{b \in \mathcal{N}_a} \sum_c \psi_{abc}^{(k, l_o, l_f, l_i, p_i)} \; Y_{m_f}^{(l_f)}(\hat{r}_{ab}) \; V_{bcm_i}^{(k-1,l_i,p_i)} \\
\text{with } \psi_{abc}^{(k,l_o, l_f, l_i, p_i)} = \Psi_c^{(k, l_o, l_f, l_i, p_i)}(e_{ab} || V_a^{(k-1, 0, 1)} || V_b^{(k-1, 0, 1)})
\end{gathered}\end{equation}\end{small}
where the outer sum is over values of $l_f, l_i, p_i$ such that $|l_i-l_f| \le l_o \le l_i+l_f$ and $(-1)^{l_f}p_i = p_o$, $C$ indicates the Clebsch-Gordan coefficients \citep{thomas2018tensor}, $\mathcal{N}_a = \{b \mid (a,b) \in \mathcal{E}_{\max} \}$ the neighborhood of $a$ and $Y$ the spherical harmonics. The rotational order of the nodes representations $l_o$ and $l_i$  and of the spherical harmonics of the edges ($l_f$) are restricted to be at most 2. All the learnable weights are contained in $\Psi$, a dictionary of MLPs that compute per-channel weights based on the edge embeddings and scalar features of the outgoing and incoming node.

\paragraph{Pseudotorque layer} The final part of our architecture is a pseudotorque layer that predicts a pseudoscalar score $\delta\tau$ for each rotatable bond from the per-node outputs of the interaction layers. For every rotatable bond, we construct a tensor-valued filter, centered on the bond, from the tensor product of the spherical harmonics with a $l=2$ representation of the \emph{bond axis}. Since the parity of the $l=2$ spherical harmonic is even, this representation does not require a choice of bond direction. The filter is then used to convolve with the representations of every neighbor on a radius graph, and the products which produce pseudoscalars are passed through odd-function (i.e., with $\tanh$ nonlinearity and no bias) dense layers (not shown in equation \ref{eq:pseudotorque}) to produce a single prediction.

For all rotatable bonds $g = (g_0, g_1) \in \mathcal{E}_{\text{rot}}$ and $b \in \mathcal{V}$, let $r_{gb}$ and $\hat{r}_{gb}$ be the magnitude and direction of the vector connecting the center of bond $g$ and $b$.
\begin{equation}\label{eq:pseudotorque}
\begin{gathered}
\mathcal{E}_\tau = \{ (g,b) \mid g \in \mathcal{E}_r, b \in \mathcal{V}, r_{gb} < r_{\max} \} \quad \quad e_{gb} = \Upsilon^{(\tau)}(\mu(r_{gb}))\\
T_{gbm_o}^{(l_o, p_o)} = \sum_{m_g, m_r, l_r: p_o=(-1)^{l_r}} C_{(2, m_g)(l_r, m_r)}^{(l_o,m_o)} Y_{m_f}^{(2)}(\hat{r}_{g}) \; Y_{m_r}^{(l_r)}(\hat{r}_{gb})\\
\delta\tau_{g} = \sum_{l, p_f, p_i: p_fp_i=-1} \sum_{m_o, m_i} C_{(l, m_f)(l, m_i)}^{(0,0)} \frac{1}{|\mathcal{N}_g|} \sum_{b \in \mathcal{N}_g} \sum_c \gamma_{gcb}^{(l, p_i)} \; T_{gbm_f}^{(l, p_f)} \; V_{bcm_i}^{(K,l,p_i)}  \\
\text{with } \gamma_{gcb}^{(l, p_i)} = \Gamma_c^{(l, p_i)}(e_{gb} || V_b^{(K, 0, 1)} || V_{g_0}^{(K, 0, 1)} + V_{g_1}^{(K, 0, 1)})
\end{gathered}\end{equation}

where $\Upsilon^{(\tau)}$ and $\Gamma$ are MLPs with learnable parameters and $\mathcal{N}_g = \{b \mid (g,b) \in \mathcal{E}_{\tau} \}$.

\section{Conformer matching} \label{app:matching}

The conformer matching procedure, summarised in Algorithm \ref{alg:matching}, proceeds as follows. For a molecule with $K$ conformers, we first generate $K$ random local structure estimates $\hat L$ from RDKit. To match with the ground truth local structures, we compute the cost of matching each true conformer $C$ with each estimate $\hat L$ (i.e. a $K\times K$ cost matrix), where the cost is the best RMSD that can be achieved by modifying the torsions of the RDKit conformer with local structure $\hat L$ to match the ground truth conformer $C$. Note that in practice, we compute an upper bound to this optimal RMSD using the fast von Mises torsion matching procedure proposed by \citet{stark2022equibind}. 

We then find an optimal matching of true conformers $C$ to local structure estimates $\hat{L}$ by solving the linear sum assignment problem over the approximate cost matrix \citep{crouse2016implementing}. Finally, for each matched pair, we find the true optimal $\hat{C}$ by running a differential evolution optimization procedure over the torsion angles \citep{mendez2021geometric}. The complete assignment resulting from the linear sum solution guarantees that there is no distributional shift in the local structures seen during training and inference. 

\begin{algorithm}[h]
\caption{Conformer matching}\label{alg:matching}
\KwIn{true conformers of $G$ $[C_1, ... C_K]$}
\KwOut{approximate conformers for training  $[\hat{C}_1, ... \hat{C}_K]$}
generate local structures $[\hat{L}_1, ... \hat{L}_K]$ with RDKit\;
\For{$(i,j)$ \textbf{in} $[1,K]\times [1,K]$}{
    $C_{\text{temp}}$ = von\_Mises\_matching($C_i$, $\hat{L}_j$)\;
    cost[i,j] = $\rmsd(C_i$, $C_{\text{temp}})$\;
}
assignment = linear\_sum\_assignment(cost)\;
\For{$i\leftarrow 1$ \KwTo $K$}{
    j = assignment[i]\;
    $\hat{C}_i$ = differential\_evolution($C_i$, $\hat{L}_j$, RMSD)\;
}
\end{algorithm}

Table \ref{tab:matching_results} shows the average RMSD between a ground truth conformer $C_i$ and its matched conformer $\hat{C}_i$. The average RMSD of 0.324 \AA\ obtained via conformer matching provides an approximate lower bound on the achievable AMR performance for methods that do not change the local structure and take those from RDKit (further discussion in Appendix \ref{app:discuss:rdkit}).

\begin{table}[h]
 \caption{Average $\rmsd(c_i, \hat{c}_i)$ achieved by different variants of conformer matching. "Original RDKit" refers to the RMSD between a random RDKit conformer and a ground truth conformer without any optimization. In "Von Mises optimization" and "Differential evolution," the torsions of the RDKit conformer are adjusted using the respective procedures, but the pairing of RDKit and ground truth conformers is still random. In "Conformer matching," the cost-minimizing assignment prior to differential evolution provides a 15\% improvement in average RMSD. The results are shown for a random 300-molecule subset of GEOM-DRUGS.}
 \label{tab:matching_results}
 \centering
 \begin{tabular}{lc} \toprule
 Matching method & RMSD (\AA) \\ \midrule
 Original RDKit           & 1.448     \\
 Von Mises optimization   & 0.728     \\
 Differential evolution   & 0.379     \\
 Conformer matching       & 0.324    \\ \bottomrule
 \end{tabular}
 \end{table}

\section{Additional discussion} \label{app:discuss}

\subsection{RDKit local structures} \label{app:discuss:rdkit}

In this section, we provide empirical justification for the claim that cheminformatics methods like RDKit already provide accurate local structures. It is well known in chemistry that bond lengths and angles take on a very narrow range of values due to strong energetic constraints. However, it is not trivial to empirically evaluate the claim due to the difficulty in defining a distance measure between a pair of local structures. In this section, we will employ two sets of observations: marginal error distributions and matched conformer RMSD.

\begin{figure}[h!]
    \centering
    \includegraphics[width=0.49\textwidth]{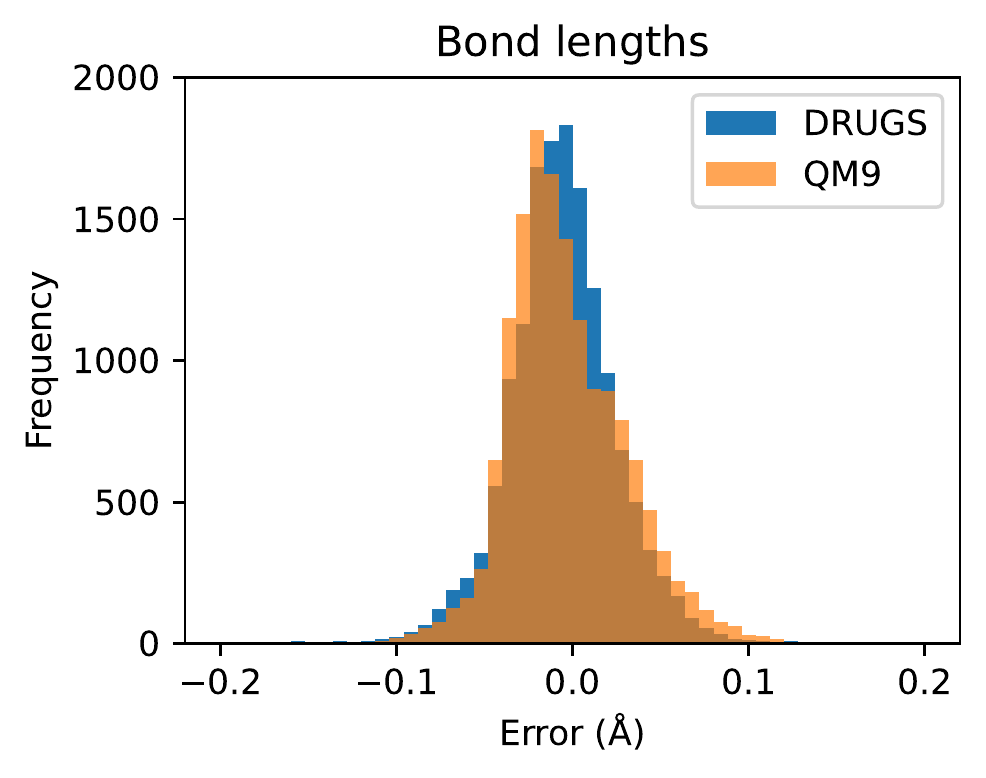}
    \includegraphics[width=0.49\textwidth]{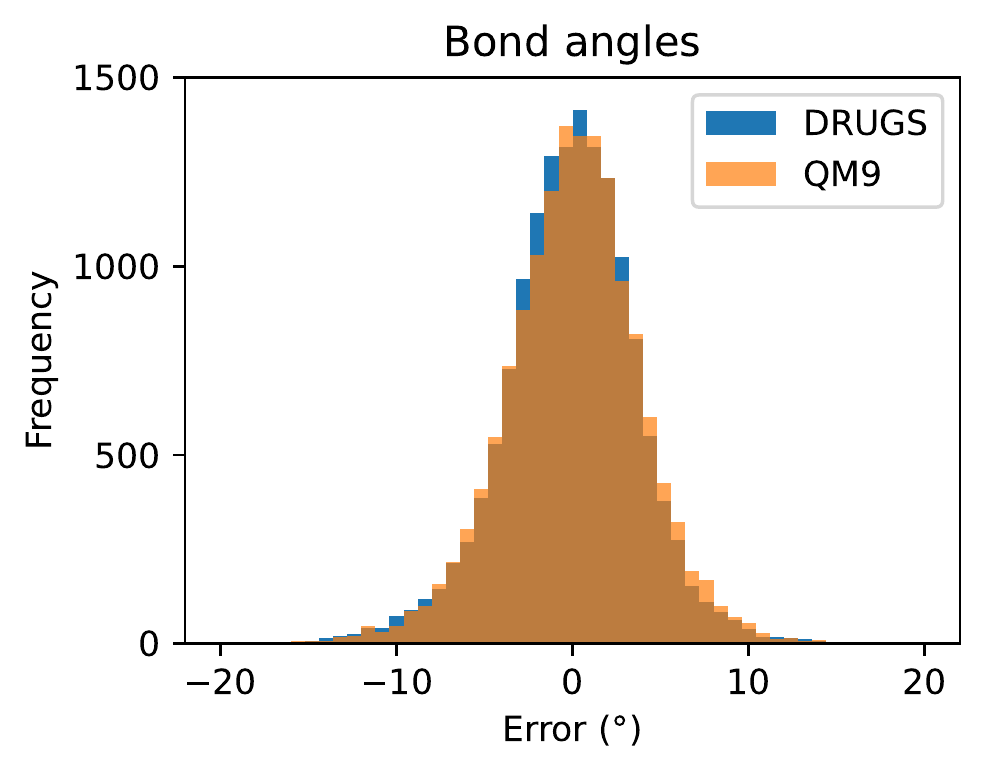}
    \caption{Histogram of the errors in 15000 predicted bond lengths and angles from randomly sampled molecules in GEOM-DRUGS and GEOM-QM9.}
    \label{fig:marginal_local_errors}
\end{figure}

\paragraph{Marginal error distributions} We examine the distribution of errors of the bond lengths and angles in a random RDKit conformer relative to the corresponding lengths and angles in a random CREST conformer (Figure \ref{fig:marginal_local_errors}). The distributions are narrow and uni-modal distributions around zero, with a RMSE of 0.03 \AA\ for bond lengths and 4.1\textdegree\, for bond angles on GEOM-DRUGS. Comparing DRUGS and QM9, the error distribution does not depend on the size of the molecule. Although it is difficult to determine how these variations will compound or compensate for each other in the global conformer structure, the analysis demonstrates that bond lengths and angles have little flexibility (i.e., no strong variability among conformers) and are accurately predicted by RDKit.

\paragraph{Matched conformer RMSD} We can more rigorously analyze the quality of a local structure $\hat{L}$ with respect to a given reference conformer $C$ by computing the minimum RMSD that can be obtained by combining $\hat{L}$ with optimal torsion angles. That is, we consider the RMSD distance of $C$ to the closest point on the manifold of possible conformers with local structure $\hat{L}$: $\rmsd_{\min}(C, \hat L) := \min_\tau \rmsd(C, \hat{C})$ where $\hat{C} = (\hat{L}, \tau)$. 

Conveniently, $\hat{C}$ is precisely the output of the differential evolution in Appendix \ref{app:matching}. Thus, the average RSMD reported in the last row of Table \ref{tab:matching_results} is the expected $\rmsd_{\min}$ of an optimal assignment of RDKit local structures to ground-truth conformers. This distance---0.324 \AA\ on GEOM-DRUGS---is significantly smaller than the error of the current state-of-the-art conformer generation methods. Further, it is only slightly larger than the average $\rmsd_{\min}$ of 0.284 \AA\ resulting from matching a ground truth conformer to the local structure of another randomly chosen ground truth conformer, which provides a measure of the variability among ground truth local structures. These observations support the claim that the accuracy of existing approaches on drug-like molecules can be significantly improved via better conditional sampling of torsion angles.

\subsection{Torsion updates}\label{app:discuss:torsion}

In the main text, we viewed updates $\Delta\tau$ as changes to a torsion angle $\tau$, and asserted that the same update applied to any torsion angle at a given bond (i.e., with any choice of reference neighbors) results in the same conformer. Given this, a potentially more intuitive presentation is to \emph{define} $\Delta\tau$ for a bond as a relative rotation around that bond, without reference to any torsion angle.

Consider a rotatable bond $(b, c)$ and the connected components $\mathcal{V}(b), \mathcal{V}(c)$ formed by removing the bond. Let $\mathbf{\hat r}_{bc} = (x_c-x_b)/|x_c-x_b|$ and similarly $\mathbf{\hat r}_{cb} = -\mathbf{\hat r}_{bc}$. Because the bond is freely rotatable, consider rotations of each side of the molecule around the bond axis given by $\mathbf{\hat r}_{bc}$. Specifically, let $\mathcal{V}(b)$ be rotated by some Euler vector $\boldsymbol{\theta}_b := \theta_b\mathbf{\hat r}_{bc}$ around $x_b$, and $\mathcal{V}(c)$ by $\boldsymbol{\theta}_c := \theta_c\mathbf{\hat r}_{bc}$ around $x_c$. Then the rotations induce a \emph{torsion update} $\Delta\tau$ if $\theta_c - \theta_b = \Delta\tau$; or equivalently
\begin{equation}
    \Delta\tau = (\boldsymbol{\theta}_c - \boldsymbol{\theta}_b) \cdot \mathbf{\hat r}_{bc} 
\end{equation}
The expression remains unchanged if we swap the indices $b, c$; thus there is no sign ambiguity. Some less formal but possibly more intuitive restatements of the sign convention are:
\begin{itemize}
    \item Looking down a bond, a \emph{positive} update is given by a CCW rotation of the nearer side; or a CW rotation of the further side
    \item For a viewer positioned in the middle of the bond, a \emph{positive} update is given by the CW rotation of any one side
    \item A \emph{positive} update is given by Euler vectors that point \emph{outwards} from the bond
\end{itemize}

These are illustrated in Figure~\ref{fig:torsion_intuition}. 

Since the Euler vector $\boldsymbol{\theta}$ is a pseudovector that remains unchanged under parity inversion, while $\mathbf{\hat r}$ is a normal vector, it is apparent that $\Delta\tau$---and any model predicting $\Delta\tau$---must be a pseudoscalar.

Because the update is determined by a {relative} rotation, it is not necessary to specify which side to rotate. That is, the same torsion update can be accomplished by rotating only one side, both sides in opposite directions, or both sides in the same direction. In practical implementation, we rotate the side of the molecule with fewer atoms, and keep the other side fixed.

\begin{figure}[h]
    \centering
    \includegraphics[width=0.6\textwidth]{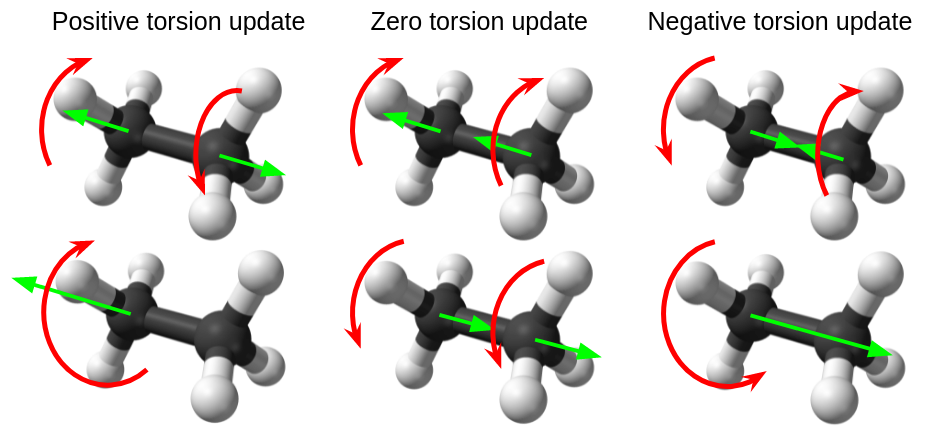}
    \caption{Torsion updates as relative rotations, with the rotations shown with curved red arrows and Euler vectors shown with straight green arrows. The second row emphasizes that the sign convention and update depend only on the \emph{relative} motion of the two sides.}
    \label{fig:torsion_intuition}
\end{figure}

\subsection{Chemical isomerism}\label{app:discuss:isomerism}

We have defined a molecule in terms of its bond connectivity, i.e., as a graph $G=(\mathcal{V}, \mathcal{E})$ with atoms $\mathcal{V}$ and $\mathcal{E}$. In chemistry, however, it is universal to consider molecules with the same connectivity, but whose conformers cannot interconvert, as {different molecules} called \emph{stereoisomers}. In our formalism, stereoisomers correspond to subsets of the space of conformers $\mathcal{C}_G$ for some molecular graph $G$. Many types of stereoisomerism exist, but the two most important are:
\begin{itemize}
    \item \textbf{Chirality}. Conformers with distinct values of chirality tags $\{z_i\}$---one for each chiral atom meeting certain criteria---are considered different molecules.
    \item \textbf{E/Z isomerism}, also called cis/trans isomerism. For each double bond meeting certain criteria, the space $[0, 2\pi)$ is partitioned into two halves, such that conformers are considered different molecules depending on the value of the torsion angle.
\end{itemize}
These are not meant to be formal definitions, and we refer to standard chemistry texts for a more detailed treatment. For our purposes, the key implication is that conformer generation requires generating conformers consistent with a \emph{given stereoisomer}. For a molecular graph $G$ with $k$ relevant chiral centers and $l$ relevant double bonds, there are $2^{k+l}$ possible stereoisomers, corresponding to the partition of $\mathcal{C}_G$ into $2^{k+l}$ disjoint subsets---only one of which corresponds to the molecule under consideration.

Torsional diffusion automatically handles chirality. Because we have considered chirality to be part of the local structure, it is drawn from the cheminformatics package RDKit, which is given the full identification of the stereoisomer along with the molecular graph, and is not modified by the torsional diffusion. Hence, our method always generates conformers with the correct chirality at each chiral center. On the other hand, GeoDiff does not consider chirality at all, while GeoMol generates molecules without any chirality constraints, and merely inverts the chiral centers that were generated incorrectly.

E/Z isomerism is significantly trickier, as it places a constraint on the torsion angles at double bonds, which are considered freely rotatable in our framework. Presently, torsional diffusion does not attempt to capture E/Z isomerism. One possible way of doing so is to augment the molecular graph with edges of a special type, and we leave such augmentation to future work. GeoDiff and GeoMol also do not attempt to treat E/Z isomerism.

More generally, while the abstract view of molecules as graphs has enabled rapid advances in molecular machine learning, stereoisomerism shows that it is clearly a simplification. As stereoisomers can have significantly different chemical properties and bioactivities, a more complete view of molecular space will be essential for further advances in molecular machine learning.

\subsection{Limitations of torsional diffusion}\label{app:discuss:limit}

As demonstrated in Section \ref{sec:experiments}, torsional diffusion significantly improves the accuracy and reduces the denoising runtime for conformer generation. However, torsional diffusion also has a number of limitations that we will discuss in this section.

\paragraph{Conformer generation} The first clear limitation is that the error that torsional diffusion can achieve is lower bounded by the \textit{quality of the local structure} from the selected cheminformatics method. As discussed in Appendix~\ref{app:discuss:rdkit}, this corresponds to the mean RMSD obtained after conformer matching, which is 0.324 \AA\ with RDKit local structures on DRUGS. Moreover, due to the the local structure \textit{distributional shift} discussed in Section \ref{sec:conf_matching}, conformer matching (or another method bridging the shift) is required to generate the training set. However, the resulting conformers are not the minima of the (unconditional or even conditional) potential energy function. Thus, the learning task becomes less physically interpretable and potentially more difficult; empirically we observe this clearly in the training and validation score-matching losses. We leave to future work the exploration of \textit{relaxations} of the rigid local structures assumption in a way that would still leverage the predominance of torsional flexibility in molecular structures, while at the same time allowing some flexibility in the independent components.

\newpage
\paragraph{Rings} The largest source of flexibility in molecular conformations that is not directly accounted for by torsional diffusion is the variability in \textit{ring conformations}. Since the torsion angles at bonds inside cycles cannot be independently varied, our framework treats them as part of the local structure. Therefore, torsional diffusion relies on the local structure sampler $p_G(L)$ to accurately model cycle conformations. Although this is true for a large number of relatively small rings (especially aromatic ones) present in many drug-like molecules, it is less true for puckered rings, fused rings, and larger cycles. In particular, torsional diffusion does not address the longstanding difficulty that existing cheminformatics methods have with macrocycles---rings with 12 or more atoms that have found several applications in drug discovery \citep{driggers2008exploration}. We hope, however, that the idea of restricting diffusion processes to the main sources of flexibility will motivate future work to define diffusion processes over cycles conformations combined with free torsion angles.

\paragraph{Boltzmann generation} With Boltzmann generators we are typically interested in sampling the Boltzmann distribution over the entire (Euclidean) conformational space $p_G(C)$. However, the procedure detailed in Section \ref{sec:energy} generates (importance-weighted) samples from the Boltzmann distribution \textit{conditioned} on a given local structure $p_G(C \mid L)$. To importance sample from the full Boltzmann distribution $p_G(C)$, one would need a model $p_G(L)$ over local structures that also provides exact likelihoods. This is not the case with RDKit or, to the best of our knowledge, other existing models, and therefore an interesting avenue for future work.

\paragraph{Proteins} As protein conformations are often described with backbone dihedral (i.e., torsion angles), it is natural to consider whether torsional diffusion may be useful for modeling protein flexibility. However, we do not believe that the \emph{direct} application of the framework to proteins or other macromolecules is very promising. Small changes in torsional coordinates cause large displacements in distant regions of the molecule, so the influence on a torsional score is not limited to the local neighborhood of the bond. For small molecules—even the ones in GEOM-XL—this is not a problem because of their limited spatial and graph theoretic diameters. In proteins, however, the graph diameter is 3 times the sequence length and can easily reach over 1000; and interactions between distant residues are extremely important in determining the structure and constraining flexibility. Although torsional diffusion may not be the right framework for modeling proteins, we believe that similar ideas (i.e., well-chosen diffusions over the flexible degrees of freedom) could be useful for generative models of protein structure and is a promising avenue of work.

\section{Experimental details} \label{app:exp_details}

\subsection{Dataset details}

\paragraph{Splits} We follow the data processing and splits from \citet{ganea2021geomol}. The splits are random with train/validation/test of 243473/30433/1000 for GEOM-DRUGS and 106586/13323/1000 for GEOM-QM9. GEOM-XL consists of only a test split (since we do not train on it), which consists of all 102 molecules in the MoleculeNet dataset with at least 100 atoms. For all splits, the molecules whose CREST conformers all have a canonical SMILES different from the SMILES of the molecule (meaning a reacted conformer), or that cannot be handled by RDKit, are filtered out. 

\paragraph{Dataset statistics} As can be seen in Figure \ref{fig:data_stats}, the datasets differ significantly in molecule size as measured by number of atoms or rotatable bonds. Particularly significant is the domain shift between DRUGS and XL, which we leverage in our experiments by testing how well models trained on DRUGS generalize to XL.

\begin{figure}[h!]
    \centering
    \includegraphics[width=0.49\textwidth]{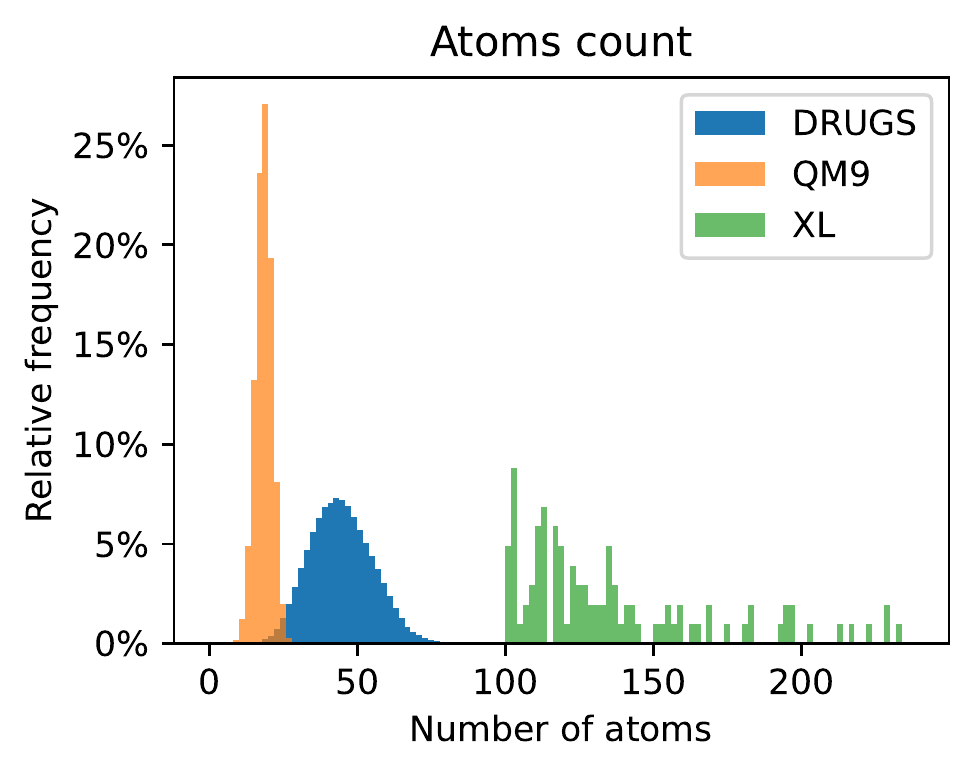}
    \includegraphics[width=0.49\textwidth]{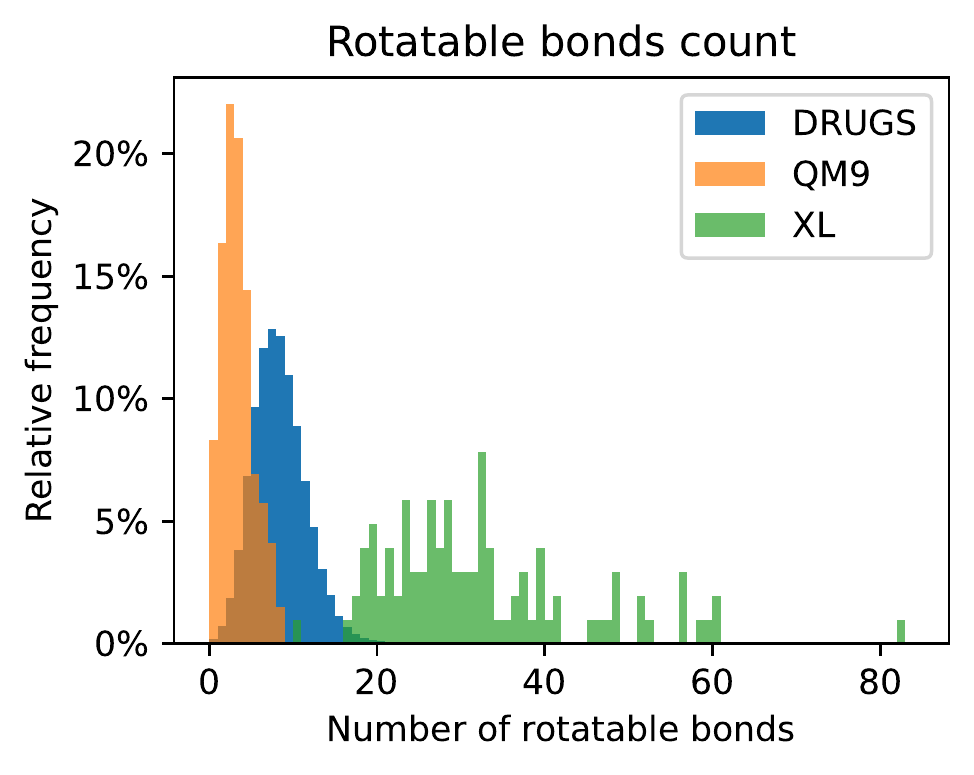}
    \caption{Statistics about the atoms and rotatable bonds counts in the three different datasets.}
    \label{fig:data_stats}
\end{figure}

\paragraph{Boltzmann generator} The torsional Boltzmann generator described in Section \ref{sec:boltzmann} is trained and tested on molecules from GEOM-DRUGS with 3--7 rotatable bonds. The training (validation) set consists of 10000 (400) such randomly selected molecules from the DRUGS training (validation) set. The test set consists of all the 453 molecules present in the DRUGS test set with 3--7 rotatable bonds.

\newpage
\subsection{Training and tuning details}

\paragraph{Conformer generation} For conformer ensemble generation on GEOM-DRUGS, the torsional diffusion models were trained on NVIDIA RTX A6000 GPUs for 250 epochs with the Adam optimizer (taking from 4 to 11 days on a single GPU). The hyperparameters tuned on the validation set were (in bold the value that was chosen): initial learning rate (0.0003, \textbf{0.001}, 0.003), learning rate scheduler patience (5, \textbf{20}), number of layers (2, \textbf{4}, 6), maximum representation order (1st, \textbf{2nd}), $r_{\max}$ (\textbf{5\AA}, 7\AA, 10\AA) and batch norm (\textbf{True}, False). All the other default hyperparameters used can be found in the attached code. For GEOM-XL the same trained model was used; for GEOM-QM9 a new model with the same hyperparameters was trained.

\paragraph{Torsional Boltzmann generators} We start from a torsional diffusion model pretrained on GEOM-DRUGS, and train for 250 epochs (6-9 days on a single GPU). A separate model is trained for every temperature. The resampling procedure with 5 steps is run for every molecule every $\max(5, ESS)$ epochs, where $ESS$ is computed for the current set of 32 samples. The only hyperparameter tuned (at temperature 300K) is $\sigma_{\min}$, the noise level at which to stop the reverse diffusion process.

We further improve the training procedure of torsional Boltzmann generators by implementing \textit{annealed training}. The Boltzmann generator for some temperature $T$ is trained at epoch $k$ by using the Boltzmann distribution at temperature $T' = T + (3000 - T)/k$ as the target distribution for that epoch. Intuitively, this trains the model at the start with a smoother distribution that is easier to learn, which gradually transforms into the desired distribution.

\subsection{Evaluation details}

\paragraph{Ensemble RMSD} As evaluation metrics for conformer generation, \cite{ganea2021geomol} and following works have used the so-called Average Minimum RMSD (AMR) and Coverage (COV) for Precision (P) and Recall (R) measured when generating twice as many conformers as provided by CREST. For $K=2L$ let $\{C^*_l\}_{l \in [1, L]}$ and $\{C_k\}_{k \in [1, K]}$ be respectively the sets of ground truth and generated conformers:
\begin{equation}
\begin{aligned}
\text{COV-R} &:= \frac{1}{L}\, \bigg\lvert\{ l \in [1..L]: \exists k \in [1..K], \rmsd(C_k, C^*_l) < \delta \, \bigg\rvert\\
\text{AMR-R} &:= \frac{1}{L} \sum_{l \in [1..L]}  \min_{ k\in [1..K]} \rmsd(C_k, C^*_l)
\end{aligned}
\end{equation}

where $\delta$ is the coverage threshold. The precision metrics are obtained by swapping ground truth and generated conformers.

In the XL dataset, due to the size of the molecules, we compute the RMSDs without testing all possible symmetries of the molecules, therefore the obtained RMSDs are an upper bound, which we find to be very close in practice to the permutation-aware RSMDs.

\paragraph{Runtime evaluation} We benchmark the methods on CPU (Intel i9-9920X) to enable comparison with RDKit. The number of threads for RDKit, \verb|numpy|, and \verb|torch| is set to 8. We select 10 molecules at random from the GEOM-DRUGS test set and generate 8 conformers per molecule using each method. Script loading and model loading times are not included in the reported values.

\paragraph{Boltzmann generator} To evaluate how well the torsional Boltzmann generator and the AIS baselines sample from the conditional Boltzmann distribution, we report their median effective sample size (ESS) \citep{kish1965survey} given the importance sampling weights $w_i$ of 32 samples for each molecule:
\begin{equation}
ESS = \frac{\big( \sum_{i=1}^{32} w_i \big)^2}{\sum_{i=1}^{32} w_i^2}
\end{equation}
This approximates the number of independent samples that would be needed from the target Boltzmann distribution to obtain an estimate with the same variance as the one obtained with the importance-weighted samples.

For the baseline annealed importance samplers, the transition kernel is a single Metropolis-Hastings step with the wrapped normal distributions on $\mathbb{T}^m$ as the proposal. We run with a range of kernel variances: $0.25, 0.5, 0.3, 0.5, 0.75, 1, 1.5. 2$; and report the best result. We use an exponential annealing schedule; i.e., $p_n \propto p_0^{1-n/N}p_N^{n/N}$ where $p_0$ is the uniform distribution and $p_N$ is the target Boltzmann density.

\section{Additional results} \label{app:results}

\paragraph{Performance vs size} Figure \ref{fig:perf_rotatable_bonds} shows the performance of different models as a function of the number of rotatable bonds. Molecules with more rotatable bonds are more flexible and are generally larger; it is therefore expected that the RMSD error will increase with the number of bonds. With very few rotatable bonds, the error of torsional diffusion depends mostly on the quality of the local structures it was given, and therefore it has a similar error as RDKit. However, as the number of torsion angles increases, torsional diffusion deteriorates more slowly than other methods. 

The trend continues with the very large molecules in GEOM-XL (average 136 atoms and 32 rotatable bonds). These not only are larger and more flexible, but---for machine learning models trained on GEOM-DRUGS---are also out of distribution. As shown in Table \ref{tab:results_xl}, on GEOM-XL GeoMol only performs marginally better than RDKit, while torsional diffusion reduces RDKit AMR by 30\% on recall and 12\% on precision. These results can very likely be improved by training and tuning the torsional diffusion model on larger molecules.

\begin{figure}[h!]
    \centering
    \includegraphics[width=0.49\textwidth]{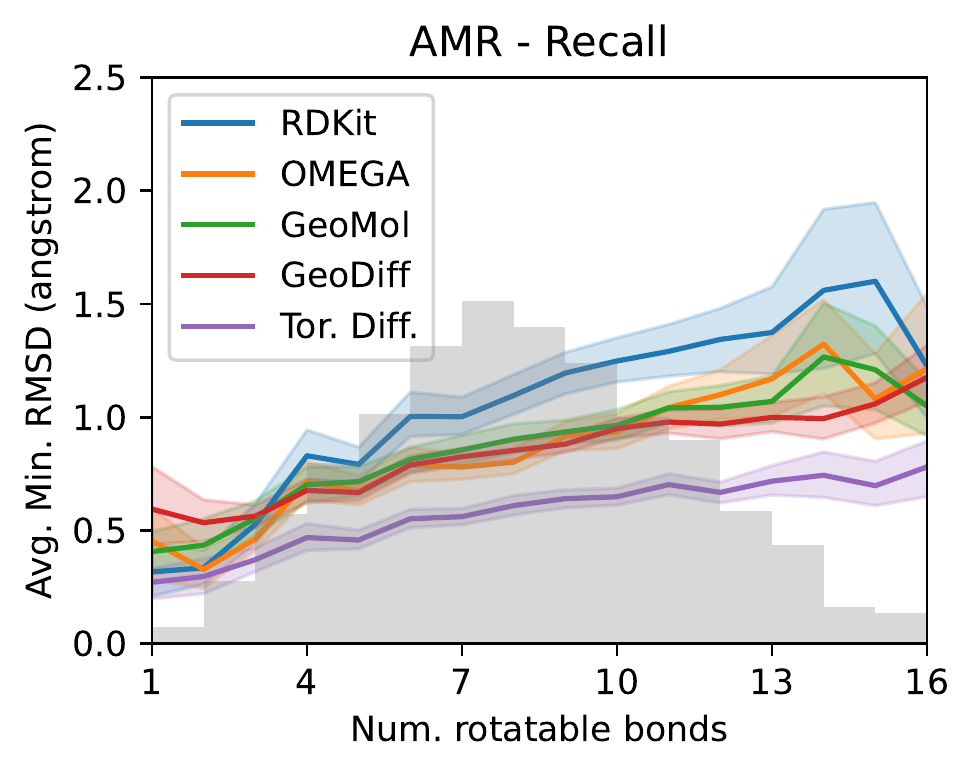}
    \includegraphics[width=0.49\textwidth]{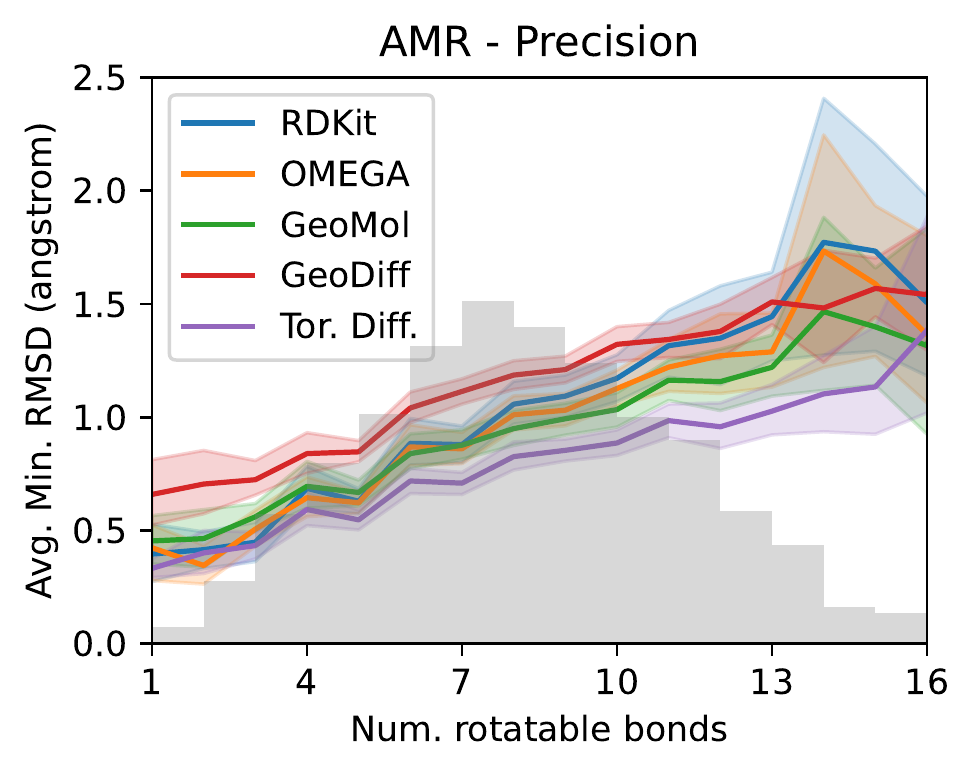}
    \caption{Average minimum RMSD (AMR) for recall (\emph{left}) and precision (\emph{right}) of the different conformer generation methods for molecules with different number of rotatable bonds in GEOM-DRUGS. The background shows the distribution of the number of rotatable bonds.}
    \label{fig:perf_rotatable_bonds}
\end{figure}

\begin{table}[ht]
\caption{Performance of various methods on the GEOM-XL dataset. } \label{tab:results_xl}
\centering
\begin{tabular}{lcccc} \toprule
                  & \multicolumn{2}{c}{AMR-R $\downarrow$} & \multicolumn{2}{c}{AMR-P $\downarrow$} \\
Model  & Mean      & Med         & Mean        & Med         \\ \midrule
RDKit               & 2.92          & 2.62          & 3.35          & 3.15          \\
GeoMol              & 2.47          & 2.39          & 3.30          & 3.15          \\
Torsional Diffusion & \textbf{2.05} & \textbf{1.86} & \textbf{2.94} & \textbf{2.78}             \\ \bottomrule
\end{tabular}
\end{table}

\newpage
\paragraph{Small molecules} We also train and evaluate our model on the small molecules from GEOM-QM9 and report the performance in Table \ref{tab:results_qm9}. For these smaller molecules, cheminformatics methods already do very well and, given the very little flexibility and few rotatable bonds present, the accuracy of local structure significantly impacts the performance of torsional diffusion. RDKit achieves a mean recall AMR just over 0.23\AA, while torsional diffusion based on RDKit local structures results in a mean recall AMR of 0.178\AA. This is already very close lower bound of 0.17\AA\ that can be achieved with RDKit local structures (as approximately calculated by conformer matching). Torsional diffusion does significantly better than other ML methods, but is only on par with or slightly worse than OMEGA, which, evidently, has a better local structures for these small molecules.

\begin{table}[t!]
\caption{Performance of various methods on the GEOM-QM9 dataset test-set ($\delta=0.5$\AA). Again GeoDiff was retrained on the splits from \cite{ganea2021geomol}. } \label{tab:results_qm9}
\begin{tabular}{l|cccc|cccc} \toprule
                & \multicolumn{4}{c|}{Recall} & \multicolumn{4}{c}{Precision}  \\
                  & \multicolumn{2}{c}{Coverage $\uparrow$} & \multicolumn{2}{c|}{AMR $\downarrow$} & \multicolumn{2}{c}{Coverage $\uparrow$} & \multicolumn{2}{c}{AMR $\downarrow$} \\
Method & Mean & Med & Mean & Med & Mean & Med & Mean & Med \\ \midrule
RDKit            & 85.1          & \textbf{100.0} & 0.235          & 0.199          & 86.8          & \textbf{100.0} & 0.232          & 0.205          \\
OMEGA            & 85.5          & \textbf{100.0} & \textbf{0.177} & \textbf{0.126} & 82.9          & \textbf{100.0} & 0.224          & \textbf{0.186} \\
GeoMol           & 91.5          & \textbf{100.0} & 0.225          & 0.193          & 86.7          & \textbf{100.0} & 0.270          & 0.241          \\
GeoDiff          & 76.5          & \textbf{100.0} & 0.297          & 0.229          & 50.0          & 33.5           & 0.524          & 0.510          \\ \midrule
Torsional diffusion             & \textbf{92.8} & \textbf{100.0} & 0.178          & 0.147          & \textbf{92.7} & \textbf{100.0} & \textbf{0.221} & 0.195             \\ \bottomrule
\end{tabular}
\end{table}

\paragraph{Ablation experiments}
In Table \ref{tab:ablations} we present a set of ablation studies to evaluate the importance of different components of the proposed torsional diffusion method:
\begin{enumerate}
    \item \textit{Baseline} refers to the model described and tested throughout the paper.
    \item \textit{Probability flow ODE} refers to using the ODE formulation of the reverse diffusion process (not an ablation, strictly speaking). As expected, it obtains similar results to the baseline SDE formulation.
    \item \textit{Only D.E. matching} refers to a model trained on conformers obtained by a random assignment of RDKit local structures to ground truth conformers (without first doing an optimal assignment as in Appendix \ref{app:matching}); this performs only marginally worse than full conformer matching.
    \item \textit{First order irreps} refers to the same model but with node irreducible representations kept only until order $\ell = 1$ instead of $\ell = 2$; this worsens the average error by about 5\%, but results in a 41\% runtime speed-up. 
    \item \textit{Train on ground truth L} refers to a model trained directly on the ground truth conformers without {conformer matching} but tested (as always) on RDKit local structures; although the training and validation score matching loss of this model is significantly lower, its inference performance reflects the detrimental effect of the local structure distributional shift.
    \item \textit{No parity equivariance} refers to a model whose outputs are parity invariant instead of parity equivariant; the model cannot distinguish a molecule from its mirror image and fails to learn, resulting in performance on par with a random baseline.
    \item \textit{Random $\boldsymbol{\tau}$} refers to a random baseline using RDKit local structures and uniformly random torsion angles.
\end{enumerate}

\newpage

\begin{table}[p]
\caption{Ablation studies with ensemble RMSD on GEOM-DRUGS. Refer to the text in the Appendix for an explanation of each entry. As usual, we compute Coverage with $\delta = 0.75$ \AA.} \label{tab:ablations}
\centering
\begin{tabular}{l|cccc|cccc} \toprule
                & \multicolumn{4}{c|}{Recall} & \multicolumn{4}{c}{Precision}  \\
                  & \multicolumn{2}{c}{Coverage $\uparrow$} & \multicolumn{2}{c|}{AMR $\downarrow$} & \multicolumn{2}{c}{Coverage $\uparrow$} & \multicolumn{2}{c}{AMR $\downarrow$} \\
Method & Mean & Med & Mean & Med & Mean & Med & Mean & Med \\ \midrule
Baseline & {72.7} & {80.0} & {0.582} & {0.565} & {55.2} & \textbf{56.9} & \textbf{0.778} & \textbf{0.729}    \\ \midrule
Probability flow ODE & \textbf{73.1} & 80.4 & \textbf{0.577} & \textbf{0.557} & \textbf{55.3} & 55.7 & 0.779 & 0.737 \\
Only D.E. matching & 72.5 & \textbf{81.1} & 0.588 & 0.569 & 53.8 & 56.1 & 0.794 & 0.749 \\
First order irreps & 70.1 & 77.9 & 0.605 & 0.589 & 51.4 & 51.4 & 0.817 & 0.783\\
Train on ground truth $L$ & 34.8 & 22.4 & 0.920 & 0.909 & 22.3 & 7.8 & 1.182 & 1.136\\
No parity equivariance & 30.5 & 12.5 & 0.928 & 0.929 & 17.9 & 3.9 & 1.234 & 1.217\\
\midrule
Random $\boldsymbol{\tau}$ & 30.9 & 13.2 & 0.922 & 0.923 & 18.2 & 4.0 & 1.228 & 1.217\\
\bottomrule
\end{tabular}
\end{table}

\begin{table}[p]
\caption{Ensemble RMSD results on GEOM-DRUGS for varying number of diffusion steps. 20 steps were used for all results reported elsewhere. As usual, we compute Coverage with $\delta = 0.75$ \AA.} \label{tab:steps}
\centering
\begin{tabular}{l|cccc|cccc} \toprule
                & \multicolumn{4}{c|}{Recall} & \multicolumn{4}{c}{Precision}  \\
                  & \multicolumn{2}{c}{Coverage $\uparrow$} & \multicolumn{2}{c|}{AMR $\downarrow$} & \multicolumn{2}{c}{Coverage $\uparrow$} & \multicolumn{2}{c}{AMR $\downarrow$} \\
Steps & Mean & Med & Mean & Med & Mean & Med & Mean & Med \\ \midrule
3 & 42.9 &	33.8	& 0.820	& 0.821 &	24.1 & 	11.1	& 1.116& 	1.100\\
5 & 58.9 &	63.6	& 0.698 & 	0.685 &	35.8 &	26.6 &	0.979	& 0.963\\
10 & 70.6 &	78.8 &	0.600 &	0.580	& 50.2 & 	48.3 &	0.827 &	0.791 \\
20 & {72.7} & {80.0} & {0.582} & {0.565} & {55.2} & {56.9} & {0.778} & {0.729}    \\ 
50 & 73.1 & 80.4 & 0.578 & 0.557 & 57.6 & 60.7 & 0.753 & 0.699\\\bottomrule
\end{tabular}
\end{table}

\begin{table}[p]
    \caption{Median absolute error of generated v.s. ground truth ensemble properties with and without relaxation. $E, \Delta\epsilon, E_{\min}$ in kcal/mol, $\mu$ in debye.}\label{tab:properties_appendix}
    \centering
    \begin{tabular}{l|cccc|cccc}
    \toprule 
    & \multicolumn{4}{c|}{Without relaxation} &  \multicolumn{4}{c}{With relaxation} \\
    Method & $E$ & $\mu$ & $\Delta \epsilon$ & $E_{\min}$ & $E$ & $\mu$ & $\Delta \epsilon$ & $E_{\min}$ \\\midrule
    RDKit & 39.08 & 1.40 & 5.04 & 39.14 & 0.81 & 0.52 & 0.75 & 1.16 \\
    OMEGA & \textbf{16.47} & \textbf{0.78} & \textbf{3.25} & \textbf{16.45} & 0.68 & 0.66 & 0.68 & 0.69 \\
    GeoMol & 43.27 & 1.22 & 7.36 & 43.68 & 0.42 & \textbf{0.34} & 0.59 & 0.40 \\
    GeoDiff & 18.82 & 1.34 & 4.96 & 19.43 & 0.31 & {0.35} & 0.89 & 0.39 \\
    Tor. Diff. & 36.91 & 0.92 & 4.93 & 36.94 & \textbf{0.22} & {0.35} & \textbf{0.54} & \textbf{0.13} \\ \bottomrule
    \end{tabular}
\end{table}

\clearpage

\paragraph{Reverse diffusion steps} In Table~\ref{tab:steps} we vary the number of steps used in the reverse diffusion process and evaluate the ensemble RMSD results on GEOM-DRUGS. We find that torsional diffusion is remarkably parsimonious in terms of number of steps required: the majority of gain in performance over prior diffusion-based methods is attained with only 10 steps. We confirm that increasing the number of steps from the default of 20 to 50 only results in minor performance gains.

\paragraph{Ensemble properties}

In Table \ref{tab:properties_appendix}, we report the median absolute errors of the Boltzmann-weighted properties of the generated vs CREST ensembles, with and without GFN2-xTB relaxation. For all methods, the errors without relaxation are far too large for the computed properties to be chemically useful---for reference, the thermal energy at room temperature is 0.59 kcal/mol. In realistic settings, relaxation of local structures is necessary for any method, after which errors from global flexibility become important. After relaxation, torsional diffusion obtains property approximations on par or better than all previous methods.

\paragraph{Torsional Boltzmann generator}

Figure \ref{fig:hist_bg} shows the histograms of ESSs at 500K for the torsional Boltzmann generator and the AIS baseline. While AIS fails to generate more than one effective sample for most molecules (tall leftmost column), torsional Boltzmann generators are much more efficient, with more than five effective samples for a significant fraction of molecules.

\begin{figure}[h]
    \centering
    \includegraphics[width=0.49\textwidth]{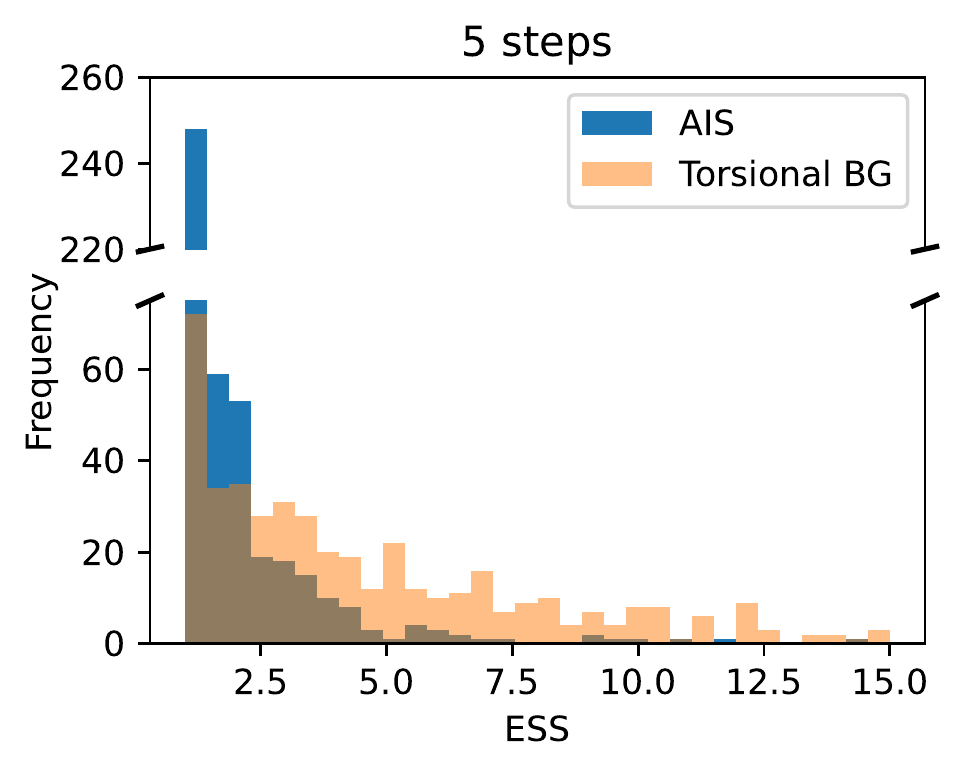}
    \includegraphics[width=0.49\textwidth]{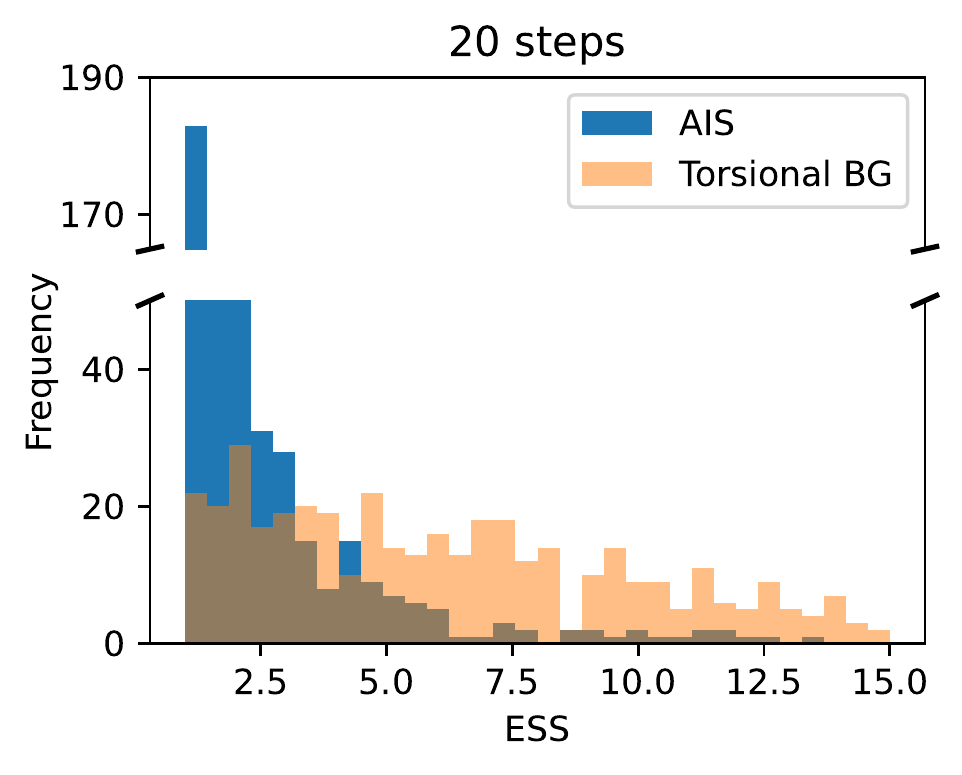}
    \caption{Histogram of the ESSs of the torsional Boltzmann generator and AIS baseline at 500K.}
    \label{fig:hist_bg}
\end{figure}

\end{document}

%% file: math_commands.tex
\usepackage{amsmath,amsfonts,bm}

\def\eqref#1{equation~\ref{#1}}

\def\1{\bm{1}}

\DeclareMathAlphabet{\mathsfit}{\encodingdefault}{\sfdefault}{m}{sl}
\SetMathAlphabet{\mathsfit}{bold}{\encodingdefault}{\sfdefault}{bx}{n}

\DeclareMathOperator{\sign}{sign}

%% file: main.bbl
\begin{thebibliography}{49}
\providecommand{\natexlab}[1]{#1}
\providecommand{\url}[1]{\texttt{#1}}
\expandafter\ifx\csname urlstyle\endcsname\relax
  \providecommand{\doi}[1]{doi: #1}\else
  \providecommand{\doi}{doi: \begingroup \urlstyle{rm}\Url}\fi

\bibitem[Anderson(1982)]{anderson1982reverse}
Brian~DO Anderson.
\newblock Reverse-time diffusion equation models.
\newblock \emph{Stochastic Processes and their Applications}, 1982.

\bibitem[Axelrod and G{\'o}mez-Bombarelli(2022)]{axelrod2020geom}
Simon Axelrod and Rafael G{\'o}mez-Bombarelli.
\newblock Geom, energy-annotated molecular conformations for property
  prediction and molecular generation.
\newblock \emph{Scientific Data}, 2022.

\bibitem[Bannwarth et~al.(2019)Bannwarth, Ehlert, and
  Grimme]{bannwarth2019gfn2}
Christoph Bannwarth, Sebastian Ehlert, and Stefan Grimme.
\newblock Gfn2-xtb—an accurate and broadly parametrized self-consistent
  tight-binding quantum chemical method with multipole electrostatics and
  density-dependent dispersion contributions.
\newblock \emph{Journal of chemical theory and computation}, 2019.

\bibitem[Batzner et~al.(2022)Batzner, Musaelian, Sun, Geiger, Mailoa,
  Kornbluth, Molinari, Smidt, and Kozinsky]{batzner2021se}
Simon Batzner, Albert Musaelian, Lixin Sun, Mario Geiger, Jonathan~P Mailoa,
  Mordechai Kornbluth, Nicola Molinari, Tess~E Smidt, and Boris Kozinsky.
\newblock E (3)-equivariant graph neural networks for data-efficient and
  accurate interatomic potentials.
\newblock \emph{Nature communications}, 2022.

\bibitem[Bolton et~al.(2011)Bolton, Kim, and Bryant]{bolton2011pubchem3d}
Evan~E Bolton, Sunghwan Kim, and Stephen~H Bryant.
\newblock Pubchem3d: conformer generation.
\newblock \emph{Journal of cheminformatics}, 2011.

\bibitem[Carroll(2019)]{carroll2019spacetime}
Sean~M Carroll.
\newblock \emph{Spacetime and geometry}.
\newblock Cambridge University Press, 2019.

\bibitem[Chirikjian(2011)]{chirikjian2011stochastic}
Gregory~S Chirikjian.
\newblock \emph{Stochastic models, information theory, and Lie groups, volume
  2: Analytic methods and modern applications}, volume~2.
\newblock Springer Science \& Business Media, 2011.

\bibitem[Cole et~al.(2018)Cole, Korb, McCabe, Read, and
  Taylor]{cole2018knowledge}
Jason~C Cole, Oliver Korb, Patrick McCabe, Murray~G Read, and Robin Taylor.
\newblock Knowledge-based conformer generation using the cambridge structural
  database.
\newblock \emph{Journal of Chemical Information and Modeling}, 2018.

\bibitem[Crouse(2016)]{crouse2016implementing}
David~F Crouse.
\newblock On implementing 2d rectangular assignment algorithms.
\newblock \emph{IEEE Transactions on Aerospace and Electronic Systems}, 2016.

\bibitem[De~Bortoli et~al.(2022)De~Bortoli, Mathieu, Hutchinson, Thornton, Teh,
  and Doucet]{de2022riemannian}
Valentin De~Bortoli, Emile Mathieu, Michael Hutchinson, James Thornton,
  Yee~Whye Teh, and Arnaud Doucet.
\newblock Riemannian score-based generative modeling.
\newblock \emph{arXiv preprint}, 2022.

\bibitem[Driggers et~al.(2008)Driggers, Hale, Lee, and
  Terrett]{driggers2008exploration}
Edward~M Driggers, Stephen~P Hale, Jinbo Lee, and Nicholas~K Terrett.
\newblock The exploration of macrocycles for drug discovery—an underexploited
  structural class.
\newblock \emph{Nature Reviews Drug Discovery}, 2008.

\bibitem[Ganea et~al.(2021)Ganea, Pattanaik, Coley, Barzilay, Jensen, Green,
  and Jaakkola]{ganea2021geomol}
Octavian Ganea, Lagnajit Pattanaik, Connor Coley, Regina Barzilay, Klavs
  Jensen, William Green, and Tommi Jaakkola.
\newblock Geomol: Torsional geometric generation of molecular 3d conformer
  ensembles.
\newblock \emph{Advances in Neural Information Processing Systems}, 2021.

\bibitem[Geiger et~al.(2022)Geiger, Smidt, M., Miller, Boomsma, Dice,
  Lapchevskyi, Weiler, Tyszkiewicz, Batzner, Madisetti, Uhrin, Frellsen, Jung,
  Sanborn, Wen, Rackers, Rød, and Bailey]{e3nn}
Mario Geiger, Tess Smidt, Alby M., Benjamin~Kurt Miller, Wouter Boomsma,
  Bradley Dice, Kostiantyn Lapchevskyi, Maurice Weiler, Michał Tyszkiewicz,
  Simon Batzner, Dylan Madisetti, Martin Uhrin, Jes Frellsen, Nuri Jung, Sophia
  Sanborn, Mingjian Wen, Josh Rackers, Marcel Rød, and Michael Bailey.
\newblock Euclidean neural networks: e3nn, April 2022.
\newblock URL \url{https://doi.org/10.5281/zenodo.6459381}.

\bibitem[Halgren(1996)]{halgren1996merck}
Thomas~A Halgren.
\newblock Merck molecular force field. i. basis, form, scope, parameterization,
  and performance of mmff94.
\newblock \emph{Journal of computational chemistry}, 1996.

\bibitem[Hawkins(2017)]{hawkins2017conformation}
Paul~CD Hawkins.
\newblock Conformation generation: the state of the art.
\newblock \emph{Journal of Chemical Information and Modeling}, 2017.

\bibitem[Hawkins and Nicholls(2012)]{hawkins2012conformer}
Paul~CD Hawkins and Anthony Nicholls.
\newblock Conformer generation with omega: learning from the data set and the
  analysis of failures.
\newblock \emph{Journal of chemical information and modeling}, 2012.

\bibitem[Hawkins et~al.(2010)Hawkins, Skillman, Warren, Ellingson, and
  Stahl]{hawkins2010conformer}
Paul~CD Hawkins, A~Geoffrey Skillman, Gregory~L Warren, Benjamin~A Ellingson,
  and Matthew~T Stahl.
\newblock Conformer generation with omega: algorithm and validation using high
  quality structures from the protein databank and cambridge structural
  database.
\newblock \emph{Journal of chemical information and modeling}, 2010.

\bibitem[Ho et~al.(2020)Ho, Jain, and Abbeel]{ho2020denoising}
Jonathan Ho, Ajay Jain, and Pieter Abbeel.
\newblock Denoising diffusion probabilistic models.
\newblock \emph{Advances in Neural Information Processing Systems}, 2020.

\bibitem[Hutchinson(1989)]{hutchinson1989stochastic}
Michael~F Hutchinson.
\newblock A stochastic estimator of the trace of the influence matrix for
  laplacian smoothing splines.
\newblock \emph{Communications in Statistics-Simulation and Computation}, 1989.

\bibitem[Jing et~al.(2021)Jing, Eismann, Suriana, Townshend, and
  Dror]{jing2020learning}
Bowen Jing, Stephan Eismann, Patricia Suriana, Raphael John~Lamarre Townshend,
  and Ron Dror.
\newblock Learning from protein structure with geometric vector perceptrons.
\newblock In \emph{International Conference on Learning Representations}, 2021.

\bibitem[Jing et~al.(2022)Jing, Corso, Berlinghieri, and
  Jaakkola]{jing2022subspace}
Bowen Jing, Gabriele Corso, Renato Berlinghieri, and Tommi Jaakkola.
\newblock Subspace diffusion generative models.
\newblock \emph{arXiv preprint arXiv:2205.01490}, 2022.

\bibitem[Kish(1965)]{kish1965survey}
Leslie Kish.
\newblock \emph{Survey sampling}.
\newblock Number 04; HN29, K5. 1965.

\bibitem[K{\"o}hler et~al.(2021)K{\"o}hler, Kr{\"a}mer, and
  No{\'e}]{kohler2021smooth}
Jonas K{\"o}hler, Andreas Kr{\"a}mer, and Frank No{\'e}.
\newblock Smooth normalizing flows.
\newblock \emph{Advances in Neural Information Processing Systems}, 2021.

\bibitem[Lagorce et~al.(2009)Lagorce, Pencheva, Villoutreix, and
  Miteva]{lagorce2009dg}
David Lagorce, Tania Pencheva, Bruno~O Villoutreix, and Maria~A Miteva.
\newblock Dg-ammos: A new tool to generate 3d conformation of small molecules
  using d istance g eometry and a utomated m olecular m echanics o ptimization
  for in silico s creening.
\newblock \emph{BMC Chemical Biology}, 2009.

\bibitem[Landrum et~al.(2013)]{landrum2013rdkit}
Greg Landrum et~al.
\newblock Rdkit: A software suite for cheminformatics, computational chemistry,
  and predictive modeling, 2013.

\bibitem[Li et~al.(2007)Li, Ehlers, Sutter, Varma-O'Brien, and
  Kirchmair]{li2007caesar}
Jiabo Li, Tedman Ehlers, Jon Sutter, Shikha Varma-O'Brien, and Johannes
  Kirchmair.
\newblock Caesar: a new conformer generation algorithm based on recursive
  buildup and local rotational symmetry consideration.
\newblock \emph{Journal of chemical information and modeling}, 2007.

\bibitem[Luo et~al.(2021)Luo, Shi, Xu, and Tang]{luo2021predicting}
Shitong Luo, Chence Shi, Minkai Xu, and Jian Tang.
\newblock Predicting molecular conformation via dynamic graph score matching.
\newblock \emph{Advances in Neural Information Processing Systems}, 2021.

\bibitem[M{\'e}ndez-Lucio et~al.(2021)M{\'e}ndez-Lucio, Ahmad, del Rio-Chanona,
  and Wegner]{mendez2021geometric}
Oscar M{\'e}ndez-Lucio, Mazen Ahmad, Ehecatl~Antonio del Rio-Chanona, and
  J{\"o}rg~Kurt Wegner.
\newblock A geometric deep learning approach to predict binding conformations
  of bioactive molecules.
\newblock \emph{Nature Machine Intelligence}, 2021.

\bibitem[Midgley et~al.(2021)Midgley, Stimper, Simm, and
  Hern{\'a}ndez-Lobato]{midgley2021bootstrap}
Laurence~Illing Midgley, Vincent Stimper, Gregor~NC Simm, and Jos{\'e}~Miguel
  Hern{\'a}ndez-Lobato.
\newblock Bootstrap your flow.
\newblock \emph{arXiv preprint}, 2021.

\bibitem[Miteva et~al.(2010)Miteva, Guyon, and Tuffery]{miteva2010frog2}
Maria~A Miteva, Frederic Guyon, and Pierre Tuffery.
\newblock Frog2: Efficient 3d conformation ensemble generator for small
  compounds.
\newblock \emph{Nucleic acids research}, 2010.

\bibitem[Neal(2001)]{neal2001annealed}
Radford~M Neal.
\newblock Annealed importance sampling.
\newblock \emph{Statistics and computing}, 2001.

\bibitem[Nichol and Dhariwal(2021)]{nichol2021improved}
Alex Nichol and Prafulla Dhariwal.
\newblock Improved denoising diffusion probabilistic models.
\newblock In \emph{International Conference on Machine Learning}, 2021.

\bibitem[No{\'e} et~al.(2019)No{\'e}, Olsson, K{\"o}hler, and
  Wu]{noe2019boltzmann}
Frank No{\'e}, Simon Olsson, Jonas K{\"o}hler, and Hao Wu.
\newblock Boltzmann generators: Sampling equilibrium states of many-body
  systems with deep learning.
\newblock \emph{Science}, 2019.

\bibitem[Pracht et~al.(2020)Pracht, Bohle, and Grimme]{pracht2020automated}
Philipp Pracht, Fabian Bohle, and Stefan Grimme.
\newblock Automated exploration of the low-energy chemical space with fast
  quantum chemical methods.
\newblock \emph{Physical Chemistry Chemical Physics}, 2020.

\bibitem[Quack(2002)]{quack2002important}
Martin Quack.
\newblock How important is parity violation for molecular and biomolecular
  chirality?
\newblock \emph{Angewandte Chemie International Edition}, 2002.

\bibitem[Riniker and Landrum(2015)]{riniker2015better}
Sereina Riniker and Gregory~A Landrum.
\newblock Better informed distance geometry: using what we know to improve
  conformation generation.
\newblock \emph{Journal of chemical information and modeling}, 2015.

\bibitem[Salimans and Ho(2022)]{salimans2021progressive}
Tim Salimans and Jonathan Ho.
\newblock Progressive distillation for fast sampling of diffusion models.
\newblock In \emph{International Conference on Learning Representations}, 2022.

\bibitem[Satorras et~al.(2021)Satorras, Hoogeboom, and Welling]{satorras2021n}
V{\i}ctor~Garcia Satorras, Emiel Hoogeboom, and Max Welling.
\newblock E (n) equivariant graph neural networks.
\newblock In \emph{International Conference on Machine Learning}, 2021.

\bibitem[Sch{\"u}tt et~al.(2017)Sch{\"u}tt, Kindermans, Sauceda~Felix, Chmiela,
  Tkatchenko, and M{\"u}ller]{schutt2017schnet}
Kristof Sch{\"u}tt, Pieter-Jan Kindermans, Huziel~Enoc Sauceda~Felix, Stefan
  Chmiela, Alexandre Tkatchenko, and Klaus-Robert M{\"u}ller.
\newblock Schnet: A continuous-filter convolutional neural network for modeling
  quantum interactions.
\newblock \emph{Advances in neural information processing systems}, 2017.

\bibitem[Shi et~al.(2021)Shi, Luo, Xu, and Tang]{shi2021learning}
Chence Shi, Shitong Luo, Minkai Xu, and Jian Tang.
\newblock Learning gradient fields for molecular conformation generation.
\newblock In \emph{International Conference on Machine Learning}, 2021.

\bibitem[Song and Ermon(2019)]{song2019generative}
Yang Song and Stefano Ermon.
\newblock Generative modeling by estimating gradients of the data distribution.
\newblock \emph{Advances in Neural Information Processing Systems}, 2019.

\bibitem[Song et~al.(2021)Song, Sohl-Dickstein, Kingma, Kumar, Ermon, and
  Poole]{song2021score}
Yang Song, Jascha Sohl-Dickstein, Diederik~P Kingma, Abhishek Kumar, Stefano
  Ermon, and Ben Poole.
\newblock Score-based generative modeling through stochastic differential
  equations.
\newblock \emph{International Conference on Learning Representations}, 2021.

\bibitem[St{\"a}rk et~al.(2022)St{\"a}rk, Ganea, Pattanaik, Barzilay, and
  Jaakkola]{stark2022equibind}
Hannes St{\"a}rk, Octavian-Eugen Ganea, Lagnajit Pattanaik, Regina Barzilay,
  and Tommi Jaakkola.
\newblock Equibind: Geometric deep learning for drug binding structure
  prediction.
\newblock In \emph{International Conference on Machine Learning}, 2022.

\bibitem[Thomas et~al.(2018)Thomas, Smidt, Kearnes, Yang, Li, Kohlhoff, and
  Riley]{thomas2018tensor}
Nathaniel Thomas, Tess Smidt, Steven Kearnes, Lusann Yang, Li~Li, Kai Kohlhoff,
  and Patrick Riley.
\newblock Tensor field networks: Rotation-and translation-equivariant neural
  networks for 3d point clouds.
\newblock \emph{arXiv preprint}, 2018.

\bibitem[Vahdat et~al.(2021)Vahdat, Kreis, and Kautz]{vahdat2021score}
Arash Vahdat, Karsten Kreis, and Jan Kautz.
\newblock Score-based generative modeling in latent space.
\newblock In \emph{Advances in Neural Information Processing Systems}, 2021.

\bibitem[Vaswani et~al.(2017)Vaswani, Shazeer, Parmar, Uszkoreit, Jones, Gomez,
  Kaiser, and Polosukhin]{vaswani2017attention}
Ashish Vaswani, Noam Shazeer, Niki Parmar, Jakob Uszkoreit, Llion Jones,
  Aidan~N Gomez, {\L}ukasz Kaiser, and Illia Polosukhin.
\newblock Attention is all you need.
\newblock \emph{Advances in neural information processing systems}, 2017.

\bibitem[Xu et~al.(2021{\natexlab{a}})Xu, Luo, Bengio, Peng, and
  Tang]{xu2020learning}
Minkai Xu, Shitong Luo, Yoshua Bengio, Jian Peng, and Jian Tang.
\newblock Learning neural generative dynamics for molecular conformation
  generation.
\newblock In \emph{International Conference on Learning Representations},
  2021{\natexlab{a}}.

\bibitem[Xu et~al.(2021{\natexlab{b}})Xu, Wang, Luo, Shi, Bengio,
  Gomez-Bombarelli, and Tang]{xu2021end}
Minkai Xu, Wujie Wang, Shitong Luo, Chence Shi, Yoshua Bengio, Rafael
  Gomez-Bombarelli, and Jian Tang.
\newblock An end-to-end framework for molecular conformation generation via
  bilevel programming.
\newblock In \emph{International Conference on Machine Learning},
  2021{\natexlab{b}}.

\bibitem[Xu et~al.(2022)Xu, Yu, Song, Shi, Ermon, and Tang]{xu2021geodiff}
Minkai Xu, Lantao Yu, Yang Song, Chence Shi, Stefano Ermon, and Jian Tang.
\newblock Geodiff: A geometric diffusion model for molecular conformation
  generation.
\newblock In \emph{International Conference on Learning Representations}, 2022.

\end{thebibliography}
